\documentclass[12pt]{elsarticle}

\usepackage{a4wide}
\usepackage{amssymb,amsmath,amsthm}
\usepackage{epsfig,graphicx,epstopdf}
\usepackage{hyperref,color}
\usepackage{float}

\date{\today}

\newtheorem{theorem}{Theorem}

\newtheorem{definition}[theorem]{Definition}

\newtheorem{remark}[theorem]{Remark}

\makeatletter
\@namedef{subjclassname@2020}{%
 \textup{2020} Mathematics Subject Classification}
\makeatother

\begin{document}
	
\begin{frontmatter}
		
\title{A Digital Twin of a Compartmental Epidemiological Model based on a Stieltjes Differential Equation}

\author[Add:a]{Iv\'an Area}
\ead{area@uvigo.gal}

\author[Add:b]{F.J. Fern\'andez}
\ead{fjavier.fernandez@usc.es}

\author[Add:b]{Juan J. Nieto\corref{corD}}
\ead{juanjose.nieto.roig@usc.es}
\cortext[corD]{Corresponding author.}

\author[Add:b]{F. Adrián F. Tojo}
\ead{fernandoadrian.fernandez@usc.es}

\address[Add:a]{Universidade de Vigo, Departamento de Matem\'atica Aplicada II, 
E. E. Aeron\'autica e do Espazo, Campus de Ourense, 32004 Ourense, Spain}

\address[Add:b]{Instituto de Matem\'aticas, Universidade de Santiago de Compostela, 
15782 Santiago de Compostela, Spain}


\begin{abstract}
We introduce a digital twin of the classical compartmental SIR (Susceptible, Infected, Recovered) epidemic
model and study the interrelation between the digital twin and the system. In doing so, we use Stieltjes derivatives to feed the data from the real system to the virtual model which, in return, improves it in real time. As a byproduct of the model, we present a precise mathematical definition of solution to the problem. We also analyze the existence and uniqueness of solutions, 
introduce the concept of Main Digital Twin and present some numerical simulations with real data of the COVID-19 epidemic, showing the accuracy of the proposed ideas.
\end{abstract}


\end{frontmatter}

\section{Introduction}

A digital twin (DT) can be defined as an evolving digital profile of the historical and current behavior of a physical object or real process that helps optimize the performance of the real process. It can also be defined as an avatar of a real physical system which exists in the computer \cite{Ganguli,Piascik}. It is important to warn that a DT is not just a mathematical model or a virtual counterpart. It may be a classical mathematical model, but with the plus of real-time inter-connections and feedback from the real product or system to the virtual model and vice versa. A DT is, in other words, a virtual model representing its physical counterpart which interacts with it in a continuous loop. The name comes from a Roadmap report by John Vickers of NASA in 2010 in which the authors goal was to improve the physical model simulation of spacecraft \cite{Negri,Piascik}.

A DT consists of three distinct parts, namely: the physical object or real system, the digital/virtual model, and connections between these two objects. This method has been applied to, e.g., aircraft engines \cite{Chenzhao,Tuegel}, wind turbines \cite{Tao} and networks \cite{Alam}, just to cite some of them -cf. \cite{Chakraborty}. The main idea is that the real and digital objects are related in a way such that data flows from the real system to the virtual one, and the information from the digital one is available to understand and predict the evolution of the real system, or even act on it to modify its behavior. In this way, both objects interact in manner which is twice beneficial: allowing the obtaining of a more accurate model and shaping the real objects behavior accordingly.

\medskip
The study of pandemics by using mathematical compartmental models can be traced back to Kermack and McKendrick \cite{Kermack}. The objective of this method is to divide the total population in disjoint groups or compartments, depending on the state of health, and then study the spread level of a certain disease among the population. In this setting, every individual in the same group has the same characteristics, which simplifies the study of the model. The simplest epidemiological model is usually referred as SIR model and divides the total (constant) population in three compartments, namely: Susceptible to infection (those individuals who have not been previously exposed to the disease), Infected (those individuals who have been infected by the illness) and Recovered (those who have passed the disease). Despite its simplicity, the SIR model can yield powerfull insight in what managing a pandemic concerns, so it has been applied to analyze the spread of  COVID-19, \cite{Liao} among others. That said, there are many variations of the basic SIR model, and some of those include other compartments such as exposed, hospitalized, or superspreaders \cite{Wuhan}, and which consider classical or fractional derivatives \cite{Fractional}.

\medskip
The classical derivative can be also substituted by other operators such as Stieltjes derivatives \cite{POUSO2015,FriLo17}.  This kind of derivatives (assuming that $g$ is a continuous function at $x$) can be pictured as taking the value
\[
f'_{g}(x)=\lim_{y \to x} \frac{f(y)-f(x)}{g(y)-g(x)}.
\] for some nondecreasing and left-continuous function $g: {\mathbb{R}} \to {\mathbb{R}}$. This mathematical tool possesses many of the important properties of the usual derivative. Moreover, it is extremely helpful when it comes to analyzing systems of differential equations with impulses or latency periods \cite{LoMa18}. In our case, the impulses can be understood as corrections made to the initial model in order to take into consideration the new (real) data. Thereby, link between DT and Stieltjes derivatives   appears in a natural way.

\medskip
The main aim of this work is to introduce a precise mathematical definition of a solution to a digital twin, by using the Stieltjes derivative. Existence and uniqueness of solution are obtained by considering the proposed definition. Reinterpreting the concept of the digital twin as a Stieltjes differential equation allows us to study issues such as the existence and uniqueness of the said concept from a mathematical point of view. As an example, we will consider the simple SIR model to which we will feed real COVID-19 data which will illustrate, in turn, the utility of this approach. Numerical experiments are also presented to show the proposed definition and results.

\medskip
The structure of this work is as follows: In Section \ref{sec:2} we introduce the basic definitions and notations. More precisely, we present the classical SIR model by using classical derivatives, and introduce Stieltjes derivatives. In Section \ref{sec:3}, we give a new interpretation of the digital twin(s) by using the aforementioned Stieltjes derivatives and present a definition of solution, as well as a result concerning the existence of solution to the corresponding system of differential equations. Numerical simulations are presented in Section \ref{sec:4}, comparing this approach to real data of the COVID-19 pandemic. Finally, we include the conclusions in Section \ref{sec:5}.

\section{Basic definitions and notations}\label{sec:2}

Let us consider the SIR compartmental model 
\begin{equation}\label{eq:SIR}
\begin{cases}
S'(t)=-\beta\, S(t)\, I(t), \vspace{0.2cm} \\
I'(t)=\beta\, S(t)\, I(t) - \gamma I(t), \vspace{0.2cm} \\
R'(t)=\gamma\, I(t),
\end{cases}
\end{equation}
representing the evolution of a disease at each time $t\in [0,T]$ \cite{Brauer}. This is the so-called nominal model and it is calibrated at an initial time. Here, as usual, we have considered three compartments for the population, namely: $S(t)$ denotes the population of susceptible individuals to an infectious, but not deadly, disease at time $t$; $I(t)$ stands for the population of infected individuals at time $t$; and $R(t)$ is used to represent the recovered individuals at time $t$. We have used classical derivatives in the SIR model, and we have denoted by $\beta$ the rate at which an infected person infects a susceptible person, and by $\gamma$ the rate at which infected people recover from the disease. There is no explicit solution to the SIR model \cite{Harko} and the numerical solution is usually calculated by using an appropriate numerical method requiring initial conditions for $S(t)$, $I(t)$ and $R(t)$, or by considering a power series expansion \cite{Srivastava}. We would like to recall that the definition of \emph{basic reproduction number},
\begin{equation}\label{eq:r0}
{\mathcal{R}}_{0}=\text{population} \times \frac{\beta}{\gamma},
\end{equation}
which can be understood as the expected number of secondary cases produced by a single infection in a completely susceptible population.
\medskip

We will be working with Stieltjes derivatives, so it is necessary to introduce some concepts regarding this kind of derivative. Let $I\subset\mathbb R$ be an interval and $g:I\to\mathbb R$ a left-continuous non-decreasing function. We will refer to such functions as \emph{derivators}. Denote by $\tau_u$ the usual topology of $\mathbb R$. We define the \emph{$g$-topology} $\tau_g$ \cite{FriLo17} as the family of those sets $U\subset I$ such that for every $x\in U$ there exists $\delta\in\mathbb R^+$ such that, if $y\in I$ satisfies $|g(y)-g(x)|<\delta$, then $y\in U$.

We say $f$ is \emph{$g$-continuous at $x_0\in I$} if for  every $\epsilon\in\mathbb R^+$ there exists $\delta\in\mathbb R^+$ such that $|f(y)-f(x)|<\epsilon$ if $|g(y)-g(x)|<\delta$. We say $f$ is \emph{$g$-continuous on $I$} if it is $g$-continuous at every point in $I$.  It can be checked that a function $f:(I,\tau_g)\to(\mathbb R,\tau_u)$ is continuous if and only if it is $g$-continuous \cite[Lemma 6]{MarquezAlbes2020}.

We define the set
\[C_g:=\{ t \in \mathbb R \, : \, \mbox{$g$ is constant on $(t-\varepsilon,t+\varepsilon)$ for some $\varepsilon>0$} \}.\]
We also define $D_g$ as the set of all discontinuity points of $g$. Observe that, given that $g$ is nondecreasing, we can write $D_g=\{ t \in \mathbb R  :  \Delta^+ g(t)>0\}$ where $\Delta^+ g(t)=g(t^+)-g(t)$, $t\in\mathbb R$, and $g(t^+)$ denotes the right handside limit of $g$ at $t$. 

\begin{definition}\label{Stieltjesderivative}
	Let $g:[a,b]\to\mathbb R$ be derivator and $f:[a,b]\to\mathbb R$. We define the \emph{Stieltjes derivative}, or \emph{$g$--derivative}, of $f$ at a point $t\in \mathbb R\backslash C_g$ as
	\[
	f'_g(t)=\left\{
	\begin{array}{rl} 
		\displaystyle \lim_{s \to t}\frac{f(s)-f(t)}{g(s)-g(t)},\quad & t\not\in D_{g},\vspace{0.1cm}\\
		\displaystyle\lim_{s\to t^+}\frac{f(s)-f(t)}{g(s)-g(t)},\quad & t\in D_{g},
	\end{array}
	\right.
	\]
	provided the corresponding limits exist. In that case, we say that $f$ is \emph{$g$--differentiable at $t$}.
\end{definition}
Associated to $g$ there is also a measure $\mu_g$ defined as the Lebesgue-Stieltjes measure such that $\mu_g([a,b))=g(b)-g(a)$. Obviously, for $g(t)=t$, we have $f_g'(t)=f'(t)$, the usual derivative.

\section{The digital twin of a SIR model}\label{sec:3}

Given the epidemiological compartmental model of SIR type detailed in \eqref{eq:SIR}, we can construct its digital twin as
\begin{equation}\label{eq:DT}
\begin{cases}
\displaystyle{ \frac{\partial S}{\partial t} (t,t_{s}) =-\beta(t_{s}) S(t,t_{s}) I(t,t_{s})}, 
\vspace{0.2cm} \\
\displaystyle{\frac{\partial I}{\partial t} (t,t_{s}) =\beta(t_{s}) S(t,t_{s}) I(t,t_{s}) - \gamma(t_{s}) I(t,t_{s})}, 
\vspace{0.2cm} \\ 
\displaystyle{\frac{\partial R}{\partial t}(t,t_{s}) =\gamma(t_{s}) I(t,t_{s})},
\end{cases}
\end{equation}
where $t\in [0,T]$ and $t_{s}\in [0,T]$ are, respectively, the systems time and what we shall call, for lack of a better name, \emph{``slow time''} \cite{Ganguli}. The slow time $t_{s}$ can be thought of as a time variable which is much slower than $t$, such as, for instance, the number of cycles of the epidemic, some prescribed times at which the experimental data is available, or just a different time scale of some sort which carries with it qualitative and quantitative changes in the system which we have to take into consideration. These changes are represented by the fact that $\beta(t_{s})$ and $\gamma(t_{s})$ change with $t_{s}$.

Thus, $S$, $I$, and $R$ are functions of two variables, that is, the solution will depend also on $t_{s}$, so we can write $S(t,t_{s})$, $I(t,t_{s})$, $R(t,t_{s})$. As a  consequence, the equations of the dynamical system are expressed in terms of the partial derivatives with respect to the time variable $t$. This will not imply that the SIR model \eqref{eq:SIR} has turned into a partial differential equation system since the infinitesimal dependence is restricted to the variable $t$.

For $t_{s}=0$ (the initial time) the digital twin reduces to the original model since the only information available at that time is the initial conditions. It is important to observe that the rates $\beta(t_{s})$ and $\gamma(t_{s})$ are unknown, so it will be necessary to estimate them from the data. 

It is assumed that changes in $\beta(t_s)$ and $\gamma(t_s)$ are so slow that the dynamics of \eqref{eq:DT} are effectively decoupled from these functional variations. As a consequence, both $\beta(t_s)$ and $\gamma(t_s)$ are constant as far as the instantaneous dynamics of \eqref{eq:DT}, as a function of the time $t$, is concerned. 

The system receives data and measures at times determined by $t_{s}$. We can assume that $t_s\in T_s$, where $\overline{T_s}=T_s\subset[0,T]$ is a locally finite set. Since, in this case, we are taking as time domain $[0,T]$ instead of $\mathbb R$ or $[0,\infty)$, this implies that $T_s$ is finite. We write $T_s=\{t_s^k\}_{k\geq 0}$ and assume that $t_s^0=0$ and $t_s^k<t_s^{k+1}$, for $k\geq 1$. As previously mentioned, at each time $t_s^k$, with $k\geq 1$, the system receives data that we can use to estimate the parameters $\beta(t_{s})$ and $\gamma(t_{s})$ and also to correct 
the evolution of the digital twin. For example, suppose that at each of the times $t_s^k$ the 
system receives real data from the variables $S$, $I$ and $R$:
\begin{displaymath}
\{(t_j^k,S^k_j)\}_{j=0}^M, \, \{(t_j^k,I^k_j)\}_{j=0}^M,\, \{(t_j^k,R^k_j)\}_{j=0}^M,
\end{displaymath}
where, if we assume that the times in $T_s$ are evenly spaced by a time span $\delta t>0$, 
$t_j^k=t_s^k-j\delta t$, $j=0,\ldots,M$, are the set of measurement times and 
$S^k_j$, $I^k_j$ and $R^k_j$, $j=0,\ldots,M$ are real data at times $t_j^k$, $j=0,\ldots,M$. 
In particular, $t_s^k=t_0^k$ and it may happen that $t_s^{k-1}\in \{t_j^k:\; j=0,\ldots,M\}$ if we consider 
overlap in the data that we will use to adjust the coefficients $\gamma(t_s)$ and 
$\beta(t_s)$. For instance, if the digital twin receives a data packet every two days 
made up of real data of four days back and the periodicity of the data is daily ($\delta t=1$), 
we are in the previous situation ($M+1=4$). For the first two elements of $T_s$ 
we have that:
\begin{itemize}
\item $t_s^1=4$ and $\{t_j^1\}_{j=0}^M=\{\boldsymbol{4},\boldsymbol{3},2,1\}$, 
\item $t_s^2=6$ and $\{t_j^2\}_{j=0}^M=\{6,5,\boldsymbol{4},\boldsymbol{3}\}$.
\end{itemize}
Note that in the above situation real data at times $2$ and $4$ are present in 
the data packets that are received at times $t_s^1=4$ and $t_s^2=6$. However, the 
real data is often corrected if errors are detected in the measurements. For example, 
in the context of the COVID pandemic, the government of Spain frequently corrected 
previous real data when they detected errors in the information received from the 
local governments. So, at each time step $t_j^k$ we have the opportunity of correcting 
possible errors. With these real data we can estimate the value of the parameters $\gamma(t_s^k)$ and 
$\beta(t_s^k)$:
\begin{displaymath}
\begin{array}{rcl}
\displaystyle
\gamma(t_s^k)&\equiv&
\displaystyle
\gamma\left(\{(t_j^k,S^k_j)\}_{j=0}^k, \, \{(t_j^k,I^k_j)\}_{j=0}^k,\, \{(t_j^k,R^k_j)\}_{j=0}^k\right), 
\vspace{0.2cm} \\
\beta(t_s^k)&\equiv&
\displaystyle
\beta\left(\{(t_j^k,S^k_j)\}_{j=0}^k, \, \{(t_j^k,I^k_j)\}_{j=0}^k,\, \{(t_j^k,R^k_j)\}_{j=0}^k\right).
\end{array}
\end{displaymath}
At time $t_s^0$ we assume that $\gamma(t_s^0)$ and $\beta(t_s^0)$ are known and also 
the initial data $(S_0^0,I_0^0,R_0^0)$. Additionally, in order 
to fit the digital twin to the real data at the points $t_s^k$, we will assume that
\[
S(t_s^k,t_s^k)=S^k_0, \quad
I(t_s^k,t_s^k)=I^k_0, \quad
R(t_s^k,t_s^k)=R^k_0,
\]
for $k\geq 0$. Therefore, \eqref{eq:DT} can be expressed in the following terms:
\begin{equation} \label{eq:DT2}
\begin{cases}
\displaystyle{ \frac{\partial S}{\partial t} (t,t_{s}^k) =-\beta(t_{s}^k) S(t,t_{s}^k) I(t,t_{s})}, 
\vspace{0.2cm} \\
\displaystyle{\frac{\partial I}{\partial t} (t,t_{s}^k) =\beta(t_{s}^k) S(t,t_{s}) I(t,t_{s}^k) 
- \gamma(t_{s}^k) I(t,t_{s}^k)}, 
\vspace{0.2cm} \\ 
\displaystyle{\frac{\partial R}{\partial t}(t,t_{s}^k) =\gamma(t_{s}^k) I(t,t_{s}^k)}, 
\vspace{0.2cm} \\
\displaystyle 
S(t_s^k,t_s^k)=S^k_0,\,
I(t_s^k,t_s^k)=I^k_0,\,
R(t_s^k,t_s^k)=R^k_0,\; \forall k\geq 1.
\end{cases}
\end{equation}
Next, we shall introduce the concept of solution for the previous system \eqref{eq:DT2}. On the one hand, it only makes sense to consider the solution at pairs $(t, t_s^k) $, with $t \geq t_s^k$, with $k\geq 0$. Indeed, the system receives real data at time $ t_s^k $. Following the principle of causality, it does not make sense to consider the solution at $(t,t_s^k)$ with $t<t_s^k$ for the data provided by $t_s^k$ is not yet available at time $t_k$. On the other hand, given an element $t \in [0, T]$, there exists a finite number of pairs $(t, t_s^k)$ with $ t\geq t_s^k$. Thus, we have a finite number of values for $S$, $I$ and $R$ to chose from at time $t$:
\begin{itemize}
\item If $t\in [t_s^0,t_s^1)$, we have $(S(t,t_s^0),I(t,t_s^0),R(t,t_s^0))$.
\item If $t\in [t_s^1,t_s^2)$, we have $\{(S(t,t_s^k),I(t,t_s^k),R(t,t_s^k)),\; k=0,1\}$.
\item In general, if $t\in [t_s^j,t_s^{j+1})$, we have 
$\{(S(t,t_s^k),I(t,t_s^k),R(t,t_s^k)),\; k=0,1,\ldots,j\}$.
\item If $t\in [\sup T_s,T]$, we have $\{(S(t,t_s^k),I(t,t_s^k),R(t,t_s^k)),\; k\geq 0\}$.
\end{itemize} 
As a consequence, we can define the solution to \eqref{eq:DT2} in the following way.

\begin{definition}
The solution to the problem \eqref{eq:DT2} is a multivalued map
\begin{displaymath}
DT: t\in [0,T]\rightarrow DT(t)\subset \bigsqcup_{k\ge 0} \left(\{t_s^k\}\times \mathbb{R}^3 \right),
\end{displaymath}
where $\sqcup$ denotes the disjoint union and
\begin{displaymath}
DT(t)=\left\{\begin{array}{ll}
\displaystyle \left\{
(t_k,(S(t,t_s^k),I(t,t_s^k),R(t,t_s^k))),\; k=0,1,\ldots,j
\right\}, & t \in [t_s^j,t_s^{j+1}),\; j \geq 0,\vspace{0.2cm} \\
\displaystyle \left\{
(t_k,(S(t,t_s^k),I(t,t_s^k),R(t,t_s^k))),\; k\geq 0
\right\}, & t\in [\sup T_s,T],
\end{array}\right.
\end{displaymath}
and each triple $(S(\cdot,t_s^k),I(\cdot,t_s^k),R(\cdot,t_s^k))$ is a solution to system \eqref{eq:SIR} on $[t_s^k,t_s^{k+1})$ (or on $[\sup T_s,T]$) subject to the initial conditions $S(t_s^k)=S^k$, $I(t_s^k)=I^k$, $R(t_s^k)=R^k$ and with $\beta_s$ $\gamma_s$ instead of $\beta$ and $\gamma$ respectively.
\end{definition}
\begin{remark} If we are willing to lose the information regarding which value $t_s$ generates which point in $\mathbb R^3$, we can simplify the definition to a multivalued map
	\begin{displaymath}
		DT: t\in [0,T]\rightarrow DT(t)\subset \mathbb{R}^3,
	\end{displaymath}
	where
	\begin{displaymath}
		DT(t)=\left\{\begin{array}{ll}
			\displaystyle \left\{
			(S(t,t_s^k),I(t,t_s^k),R(t,t_s^k)),\; k=0,1,\ldots,j
			\right\}, & t \in [t_s^j,t_s^{j+1}),\; j \geq 0,\vspace{0.2cm} \\
			\displaystyle \left\{
			(S(t,t_s^k),I(t,t_s^k),R(t,t_s^k)),\; k\geq 0
			\right\}, & t\in [\sup T_s,T].
		\end{array}\right.
	\end{displaymath}
	\end{remark}
%

From the previous multivalued solution, we can define the concept of the Main Digital Twin (MDT) 
as the ``best forecast'' that we can make at all times. We assume that this best forecast at a time $t$ corresponds 
to the branch associated with the closest previous $t_s^k$ element.

\begin{definition} Under the previous hypotheses, we define the Main Digital Twin (MDT) 
as follows:
\begin{multline} \label{eq:MDT0}
MDT:t\in [0,T]\rightarrow MDT(t)
\\ =\left\{
\begin{array}{ll}
\displaystyle (S(t,t_s^k),I(t,t_s^k),R(t,t_s^k)),& t\in (t_s^k,t_s^{k+1}], \vspace{0.2cm} \\
\displaystyle (S(t,\sup T_s),I(t,\sup T_s),R(t,\sup T_s)), & t \in (\sup T_s,T]. 
\end{array}\right.
\end{multline}
\end{definition}

Observe that, given an element $t\in (t_s^k,t_s^{k+1}]$ (analogous 
for $t\in (\sup T_s,T]$),
\begin{displaymath}
(S({t_s^k}^+,t_s^k),I({t_s^k}^+,t_s^k),R({t_s^k}^+,t_s^k))
:=\lim_{t\to t_s^k+} (S(t,t_s^k),I(t,t_s^k),R(t,t_s^k)) = (S_0^k,I_0^k,R_0^k), \; k \geq 1,
\end{displaymath}
and this limit is not, in general, equal to $MDT(t_s^k)=
(S(t^k,t_s^{k-1}),I(t^k,t_s^{k-1}),R(t^k,t_s^{k-1}))$. In some way, the 
main digital twin can be interpreted as the solution to a system of differential 
equations with impulses. Indeed let us define
\begin{equation}\label{eqbcs}
\begin{array}{rcl}
\beta_s : t \in [0,T]& \rightarrow & \displaystyle \beta_s (t) = \beta(t_s^0)\chi_{[t_s^0,t_s^1]}(t)+
\sum_{k\geq 1} \beta(t_s^k) \chi_{(t_s^k,t_s^{k+1}]}(t) \\ &&
\displaystyle + \beta(\sup T_s) \chi_{(\sup T_s,T]}(t), 
\vspace{0.2cm} \\
\gamma_s : t \in [0,T]& \rightarrow & \displaystyle \gamma_s (t) = \gamma(t_s^0)\chi_{[t_s^0,t_s^1]}(t)+
\sum_{k\geq 1} \gamma(t_s^k) \chi_{(t_s^k,t_s^{k+1}]}(t) \\ &&
\displaystyle+ \gamma(\sup T_s) \chi_{(\sup T_s,T]}(t).
\end{array}
\end{equation}
We have that the main digital twin $MDT(t)=(S(t),I(t),R(t))$ is the solution to
\begin{equation} \label{eq:mdt1}
\begin{cases}
S'(t)=-\beta_s(t)\, S(t)\, I(t), \vspace{0.2cm} \\
I'(t)=\beta_s(t)\, S(t)\, I(t) - \gamma_s(t)\, I(t), \vspace{0.2cm} \\
R'(t)=\gamma_s(t)\, I(t),  \; m-a.e.\; t \in [0,T]\setminus (T_s\setminus \{t_s^0\}), 
\vspace{0.2cm} \\
S(0)=S_0^0, \; I(0)=I_0^0, \; R(0)=R_0^0, 
\vspace{0.2cm} \\
S({t_s^k}^+)=S_0^k,\;
I({t_s^k}^+)=I_0^k,\; 
R({t_s^k}^+)=R_0^k,\; k \geq 1,
\end{cases}
\end{equation}
where $m$ is the Lebesgue measure on $[0,T]$. If we define 
$\mathbf{x}(t)=(x_1(t),x_2(t),x_3(t))=(S(t),I(t),R(t))$, 
\begin{displaymath}
\mathbf{I}_{t_s^k}(\mathbf{x}(t))=(S_0^k-x_1(t),I_0^k-x_2(t),R_0^k-x_3(t)),
\end{displaymath}
and
\begin{displaymath}
\mathbf{f}_s:(t,\mathbf{x})\in [0,T]\times \mathbb{R}^3 \rightarrow 
\mathbf{f}_s(t,\mathbf{x})=\begin{pmatrix}
\displaystyle -\beta_s(t)\, x_1\, x_2 \vspace{0.2cm} \\
\displaystyle \beta_s(t)\, x_1\, x_2 - \gamma_s(t)\, x_2 \vspace{0.2cm} \\
\displaystyle \gamma_s(t)\, x_2,
\end{pmatrix},
\end{displaymath}
equation \eqref{eq:mdt1} admits a classical formulation of a system of differential equations with impulses:
\begin{equation} \label{eq:mdt2}
\begin{cases}
\mathbf{x}'(t)=\mathbf{f}(t,\mathbf{x}(t)),m-a.e.\; t \in [0,T]\setminus (T_s\setminus \{t_s^0\}), 
\vspace{0.2cm} \\
\mathbf{x}(0)=(S_0^0,I_0^0, R_0^0), \vspace{0.2cm} \\
\mathbf{x}({t_s^k}^+)=\mathbf{x}(t_s^k)+\mathbf{I}_{t_s^k}(\mathbf{x}(t_s^k)), \; k\geq 1.
\end{cases}
\end{equation}
Differential equations with impulses are a special case of differential equations with Stieltjes derivatives \cite{FriLo17,POUSO2015}. Therefore, we can define the concept of solution to problem (\ref{eq:mdt2}) in terms of a system of differential equations with Stieltjes derivatives. We have the following result.


\begin{theorem} Let us assume that $\sup T_s < T$, consider the  derivator
\begin{equation} \label{eq:derivador}
g_s: t\in [0,T]\rightarrow g_s(t)= t + \sum_{\{k \geq 1:\; t_s^k<t\}} 2^{-k},
\end{equation}
and define
\begin{displaymath}
\mathbf{F}_s: (t,\mathbf{x}) \in [0,T] \times \mathbb{R}^3 \rightarrow 
\mathbf{F}_s(t,\mathbf{x})= 
\left\{ \begin{array}{ll}
\displaystyle \mathbf{f}(t,\mathbf{x}), & t \notin T_s \setminus\{t_s^0\}, \vspace{0.2cm} \\
\displaystyle \frac{\mathbf{I}_{t}(\mathbf{x})}{g_s(t^+)-g_s(t)}, & t \in T_s \setminus \{t_s^0\}.
\end{array}\right.
\end{displaymath}
Then $\mathbf{x}$ is a solution to (\ref{eq:mdt2}) if and only if $\mathbf{x}$ solves
\begin{equation}\label{eq:mdt3}
\begin{cases}
\mathbf{x}_{g_s}'(t)=\mathbf{F}_s(t,\mathbf{x}(t)),\; g_s-a.e.\; t \in [0,T], \vspace{0.2cm} \\
\mathbf{x}(0)=(S_0^0,I_0^0, R_0^0).
\end{cases}
\end{equation}
\end{theorem}

\begin{proof} The proof is a particular case of \cite[Theorem 3.43]{Marquez} which itself draws from \cite[Section 3.2]{POUSO2015}. We include an adaptation for the convenience of the reader. First of all, note that the derivator (\ref{eq:derivador}) is well defined in the sense that it is non decreasing function and left-continuous. We must observe that, given an element $t_s^k \in T_s \setminus \{t_s^0\}$, $g({t_s^k}^+)-g(t_s^k)=2^{-k}>0$.  
So let us see now each of the conditions separately:
\begin{itemize}
\item Let $\mathbf{x}$ be an element that solves (\ref{eq:mdt2}). Given an element $t\in T_s\setminus \{t_s^0\}$ we have, since $\overline{T_s}=T_s$, that:
\begin{displaymath}
\lim_{s\to t} \frac{x_j(s)-x_k(t)}{g_s(s)-g_s(t)}=\lim_{s\to t} \frac{x_j(s)-x_k(t)}{t-s},\; j=1,2,3.
\end{displaymath}
Therefore we have that $(x_j)'_g(t)$ exists if and only if $x_j'(t)$ exists and both derivatives are 
equal, so:
\begin{displaymath}
\mathbf{x}_{g_s}'(t)=\mathbf{F}(t,\mathbf{x}(t)),\; g_s-a.e\; t \in [0,T]\setminus (T_s\setminus \{t_s^0\}).
\end{displaymath}
Finally, given an element $t_s^k\in T_s\setminus \{t_s^0\}$, we have that:
\begin{displaymath}
(x_j)'_{g_s}(t_s^k) = \frac{x_j({t_s^k}^+)-x_j(t_s^k)}{g_s({t_s^k}^+)-g_s(t_s^k)}=
\frac{{I_{t_s^k}}_j(\mathbf{x}(t_s^k))}{g_s({t_s^k}^+)-g_s(t_s^k)},
\end{displaymath}
and then
\begin{displaymath}
\mathbf{x}_{g_s}'(t_s^k)=\mathbf{F}_s(t_s^k,\mathbf{x}(t_s^k)),\; \forall\, t_s^k \in 
T_s \setminus \{t_s^0\}.
\end{displaymath}
\item Let $\mathbf{x}$ be an element that solves (\ref{eq:mdt3}). Given an element $t\in [0,T]\setminus (T_s\setminus\{t_s^0\})$, we can proceed as in the previous case and we have that
\begin{displaymath}
\mathbf{x}'(t)=\mathbf{f}(t,\mathbf{x}(t)),\; m-a.e.\; t \in [0,T].
\end{displaymath}
Finally, for elements $t_s^k \in T_s\setminus \{t_s^0\}$:
\begin{displaymath}
\mathbf{x}({t_s^k}^+)=\mathbf{x}(t_s^k)+ \mathbf{x}_g'(t_s^k) \, (g({t_s^k}^+)-
g(t_s^k)) = \mathbf{x}(t_s^k)+\mathbf{I}_{t_s^k}(\mathbf{x}(t_s^k)),
\end{displaymath}
which concludes the proof.
\end{itemize}
\end{proof}

\begin{remark} Observe that:
\begin{itemize}
\item We have encoded in the derivator the times at which the main digital twin interacts 
with the real model.
\item In the term $\mathbf{F}_s$ we have coded how the real data correct 
the evolution of the main digital model. This correction is materialized in 
the modification of the parameters $\beta$ and $\gamma$, and in the correction 
of the main digital model at the times at which it interacts with the real model.
\end{itemize}
\end{remark}
 
Therefore, the concept of the main digital twin (\ref{eq:DT2}) can be understood 
in the following sense.

\begin{definition}[Main Digital Twin] We say that a function $\mathbf{x}\in 
\left[\mathcal{AC}_g([0,T])\right]^3$, with $\mathbf{x}(0)=(S_0^0,I_0^0,R_0^0)$ is 
a main digital twin of (\ref{eq:DT2}) if
\begin{equation} \label{eq:MDT4}
\mathbf{x}_{g_s}'(t)=\mathbf{F}_s(t,\mathbf{x}(t)),\; g_s-a.e.\; t \in [0,T].
\end{equation}
\end{definition}

Next we show a local result of existence and uniqueness for the main digital twin associated to (\ref{eq:DT2}) in the case $T_s$ finite. 

\begin{theorem} [Local existence and uniqueness of the main digital twin for (\ref{eq:DT2})]\label{teole} 
Assume $T_s$ finite. Then there exists $\tau \in (0,T]$ such that (\ref{eq:DT2}) has 
an unique solution in the interval 
$[0,\tau)$.
\end{theorem}

\begin{proof} The proof is a particular case of \cite[Theorem 7.4.]{FriLo17}. For $r>0$ arbitrary we shall check the hypothesis of \cite[Theorem 7.4.]{FriLo17}. In this case, the funcion $F_s$ occurring in \eqref{eq:MDT4} is given by
	\[\mathbf{F}_s(t,(x_1,x_2,x_3))=\left\{ \begin{array}{ll}
		\displaystyle \left(
			\displaystyle -\beta_s(t)\, x_1\, x_2,
			\displaystyle \beta_s(t)\, x_1\, x_2 - \gamma_s(t)\, x_2 ,
			\displaystyle \gamma_s(t)\, x_2,
	\right), & t \notin T_s \setminus\{t_s^0\}, \vspace{0.2cm} \\
		\displaystyle \frac{(S_0^k-x_1,I_0^k-x_2,R_0^k-x_3)}{g_s((t_s^k)^+)-g_s(t_s^k)}, & t_s^k \in T_s \setminus \{t_s^0\}.
	\end{array}\right.\]
\begin{itemize}
\item For every $\mathbf{x}\in \overline{B\left((S_0^0,I_0^0,R_0^0),r\right)}$, where we are 
considering the euclidean norm in $\mathbb{R}^3$, the function $\mathbf{F}_s(\cdot,\mathbf{x})$ is clearly 
Borel measurable, in particular, $g$-measurable.
\item It is clear from \eqref{eqbcs} that $\beta_s,\, \gamma_s \in  \mathcal{L}_g^1([0,T))$, so  $\mathbf{F}_s(\cdot,(S_0^0,I_0^0,R_0^0))\in \mathcal{L}_g^1([0,T))$ as well.
\item Given $\mathbf{x},\, \mathbf{y}\in  \overline{B\left((S_0^0,I_0^0,R_0^0),r\right)}$, we have for $t\in [0,T]\setminus (T_s\setminus \{t_s^0\})$: 
\begin{displaymath}
\begin{array}{rcl}
\displaystyle
|F_1(t,\mathbf{x})-F_1(t,\mathbf{y})| & \leq & 
\displaystyle |\beta_s(t)| \, \left(\|\mathbf{x}\|+\|\mathbf{y} \| \right)
\|\mathbf{x}-\mathbf{y}\|, \vspace{0.2cm} \\
\displaystyle 
|F_2(t,\mathbf{x})-F_2(t,\mathbf{y})| & \leq & 
\displaystyle \left[|\beta_s(t)| \, \left(\|\mathbf{x}\|+\|\mathbf{y} \| \right) 
+|\gamma_s(t)| \right] \|\mathbf{x}-\mathbf{y}\|, \vspace{0.2cm} \\
\displaystyle 
|F_3(t,\mathbf{x})-F_3(t,\mathbf{y})| & \leq &
\displaystyle |\gamma_s(t)|\, \|\mathbf{x}-\mathbf{y}\|.
\end{array}
\end{displaymath}
On the other hand, given an element $t_s^k\in T_s\setminus \{t_s^0\}$, 
\begin{displaymath}
| F_j(t_s^k,\mathbf{x})-F_j(t_s^k,\mathbf{y})| \leq
\frac{1}{g({t_s^k}^+)-g(t_s^k)} \|\mathbf{x}-\mathbf{y}\|. 
\end{displaymath}
Thus if we denote by:
\begin{displaymath}
L(t)=2\, r\, |\beta_s(t)|+ |\gamma_s(t)|+\frac{1}{g({t}^+)-g(t)} \chi_{T_s\setminus \{t_s^0\}}(t),
\end{displaymath}
we have that:
\begin{displaymath}
\|\mathbf{F}(t,\mathbf{x})-\mathbf{F}(t,\mathbf{y})\|\leq 
L(t)\|\mathbf{x}-\mathbf{y}\|,
\end{displaymath}
and also
\begin{displaymath}
\|L\|_{\mathcal{L}_g^1([0,T)} \leq 
2\, r \, \|\beta_s\|_{\mathcal{L}_g^1([0,T)}+
\|\gamma_s\|_{\mathcal{L}_g^1([0,T)} + \sum_{t_s^k \in T_s\setminus \{t_s^0\}} 1 <\infty.
\end{displaymath}
\end{itemize}
\end{proof}

From now on, in order to simplify the matters, we will assume that $[0,T)$ the maximal interval in which Theorem~\ref{teole} is fulfilled. 

\section{Numerical simulations}\label{sec:4}

In this section we will present the numerical results that we have obtained using real data from the evolution of the COVID in Galicia (Spain). In our case we have considered real data from the May 2, 2020 (day 1) to April 11, 2021 (day 344). It is important to mention that real data were available since March 13, 2020, however, due to the change in the data collection methodology that took place on April 29, 2020, we have considered it appropriate to only take into account data from May 2, 2020 on. In figures 
\ref{figure1A}-\ref{figure1B} the evolution of the susceptible, infectious and recovered individuals is represented in the time period described above.

\begin{figure}[H]
\centering
\includegraphics[width=0.7\linewidth]{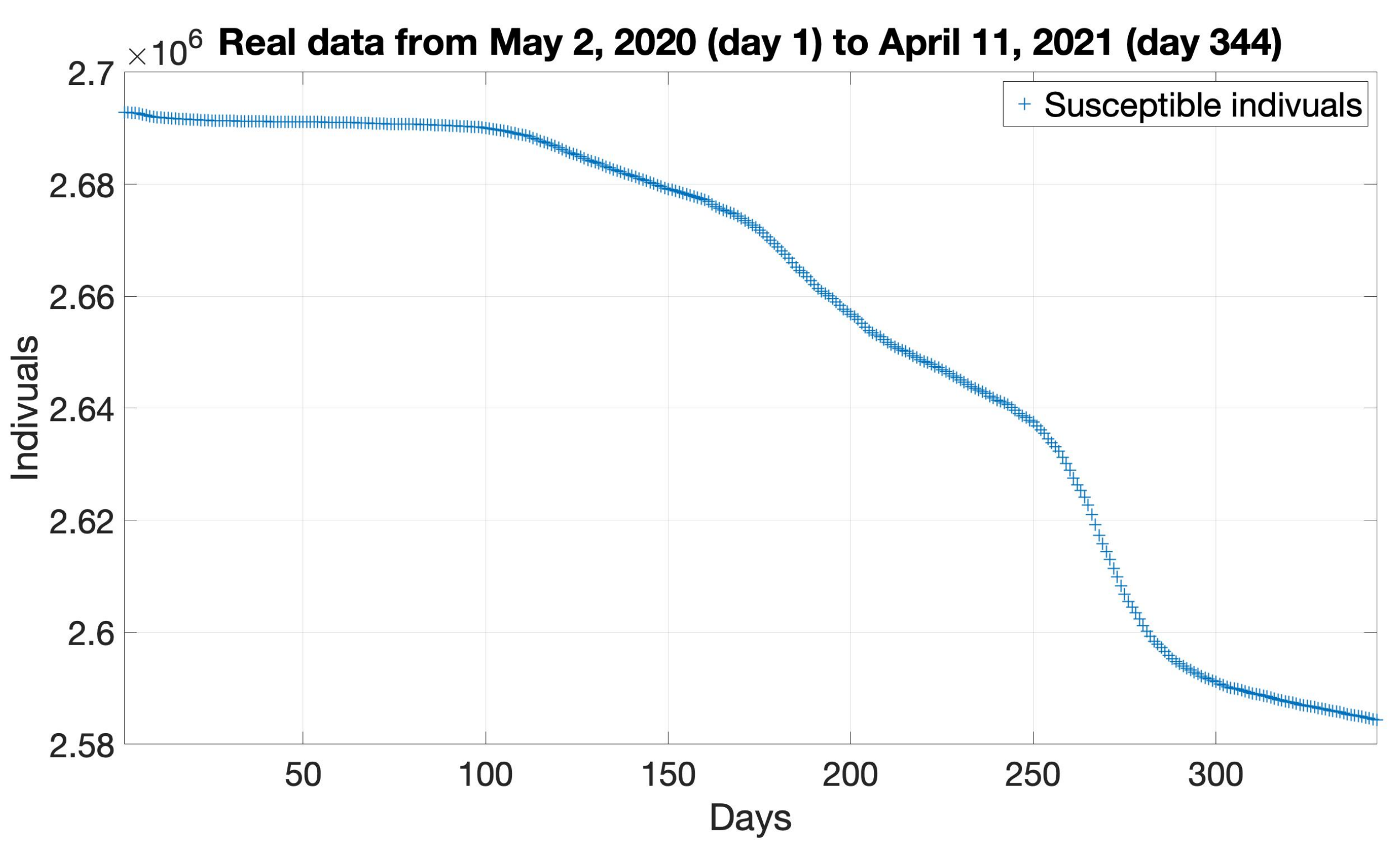}
\caption{Susceptible Individuals. Real data from May 2, 2020 (day 1) to April 11, 2021}
\label{figure1A}
\end{figure}

\begin{figure}[H]
\centering
\includegraphics[width=0.7\linewidth]{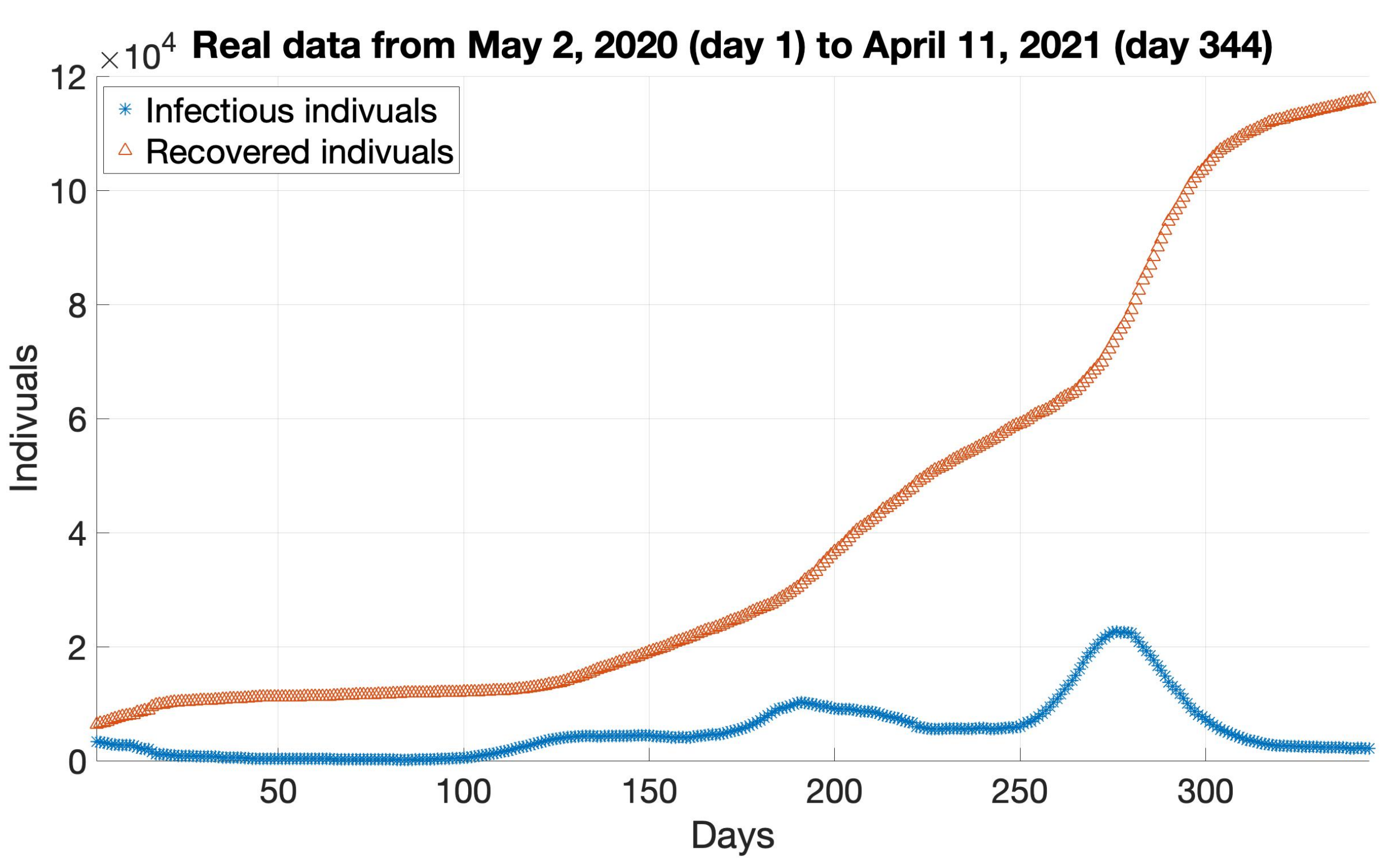}
\caption{Infectious and Recovered Individuals. Real data from May 2, 2020 (day 1) to April 11, 2021}
\label{figure1B}
\end{figure}

For the numerical simulations we have considered that the digital twin receives data 
each $7$ days and we have used $M+1=15$ days for adjusting the coefficients $\beta$ 
and $\gamma$. Thus, $t_s^0=1$, $t_s^1=15$, $t_s^k=t_s^{k-1}+7$, for $k\geq2$. The 
adjustment of the coefficients $\beta(t_s^1)$ and $\gamma(t_s^1)$ will be carried out using 
the data collected at times $\{1,\ldots,15\}$, for adjusting $\beta(t_s^2)$ and $\gamma(t_s^2)$ 
we have used data collected at times $\{8,\ldots,22\}$, and so on. 
\\

To fit the data we have used a classical fitting method based on solving the following least squares optimization problem: 
given real data $\{(t_j^k,S^k_j)\}_{j=0}^M, \, \{(t_j^k,I^k_j)\}_{j=0}^M,\, \{(t_j^k,R^k_j)\}_{j=0}^M$, find 
$(\beta(t_s^k),\gamma(t_s^k))$ such that
\begin{equation} \label{eq:fit}
J(\beta(t_s^k),\gamma(t_s^k))=\min\{J(\beta,\gamma):\; (\beta,\gamma)\in \mathbb{R}^2_+\},
\end{equation}
where
\begin{displaymath}
J(\beta,\gamma)=\sum_{j=0}^M \left(I^k_j-I(t_j^k)\right)^2+\left(R^k_j-R(t_j^k)\right)^2
\end{displaymath}
being $(S(t),I(t),R(t))$ the numerical solution to
\begin{displaymath}
\begin{cases}
S'(t)=-\beta\, S(t)\, I(t), \vspace{0.2cm} \\
I'(t)=\beta\, S(t)\, I(t) - \gamma I(t), \vspace{0.2cm} \\
R'(t)=\gamma\, I(t), \vspace{0.2cm} \\
(S(0),I(0),R(0))=(S^k_M,I^k_M,R^k_M).
\end{cases}
\end{displaymath}
To solve numerically  the previous equation we have used the Matlab command \textbf{ode45} and, for solving the optimization problem (\ref{eq:fit}), we have used the \textbf{fminsearch} Matlab function. Computations have been performed on a 2019 MacBook Pro (2,5 GHz Intel  Core i5 with 4 kernels). The CPU time required to make the $48$ adjustments was $6.88$ seconds. We observe that in the fitting process we have not used the $S$ variable due to two reasons. On the one hand, the different scale in the data could produce unsatisfactory adjustments and, on the other, the variable $S$ can be obtained from the previous ones taking into account the size of the population ($2.702.592$ individuals), so it is not necessary at the optimization step. In figures \ref{figure2A}-\ref{figure2B} and \ref{figure3A}-\ref{figure3B} 
we include two examples of the fitting process for $t_s^1=15$ and $t_s^{17}=127$, respectively.

\begin{figure}[H]
\centering
\includegraphics[width=0.7\linewidth]{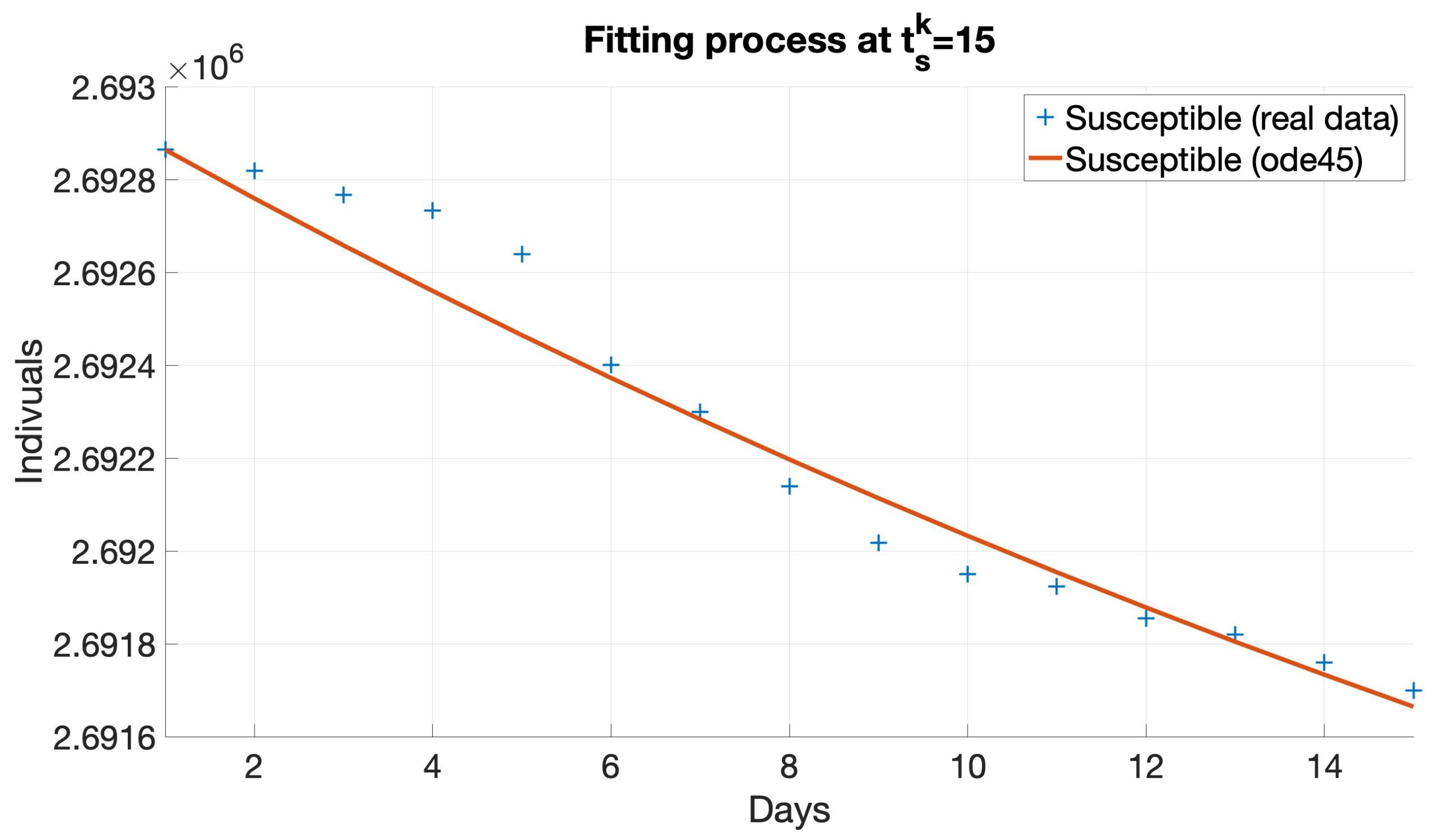}
\caption{Fitting process at $t_s^1=15$ (Susceptible Individuals)}
\label{figure2A}
\end{figure}

\begin{figure}[H]
\centering
\includegraphics[width=0.7\linewidth]{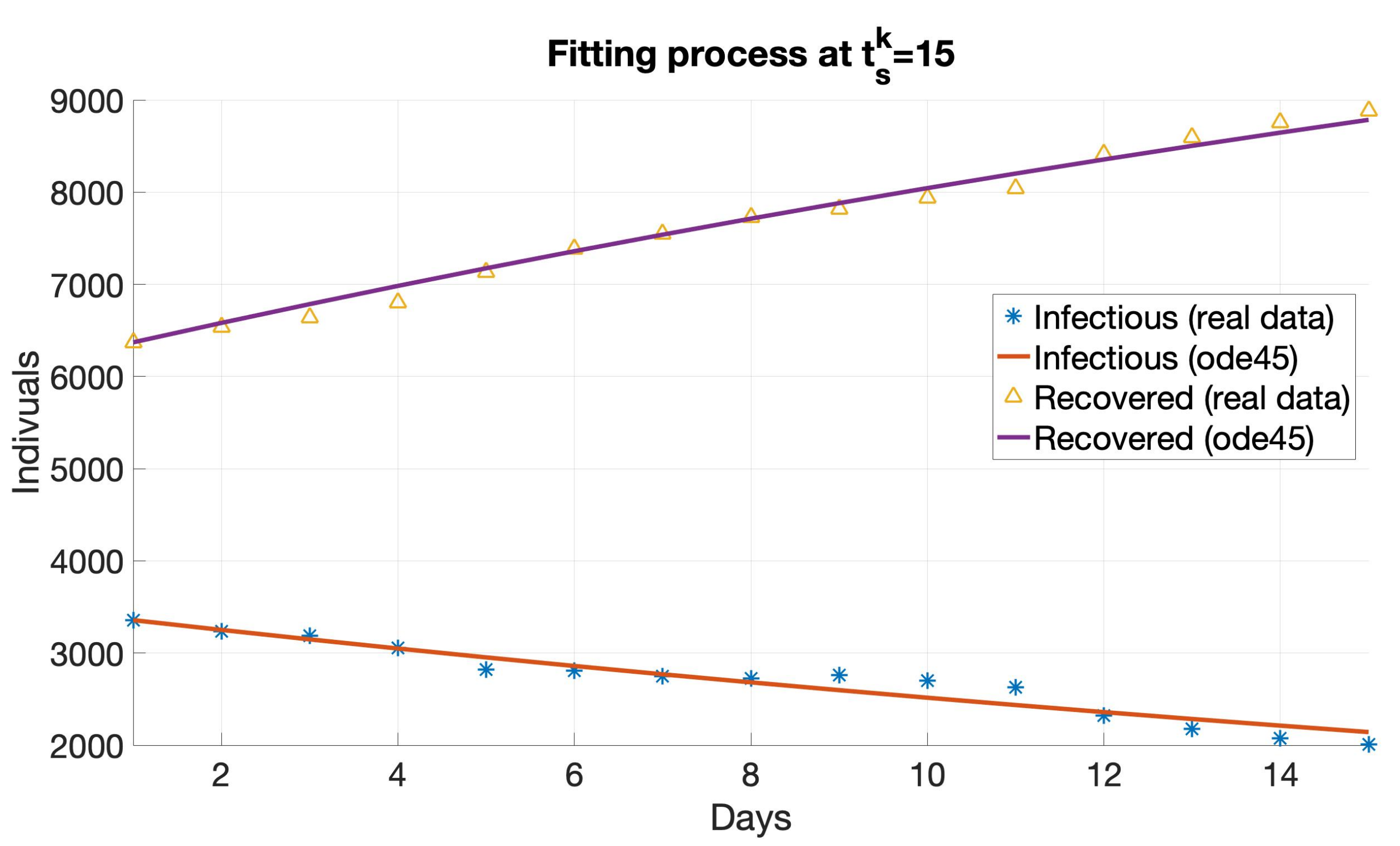}
\caption{Fitting process at $t_s^1=15$ (Infectious and Recovered Individuals)}
\label{figure2B}
\end{figure}

\begin{figure}[H]
\centering
\includegraphics[width=0.7\linewidth]{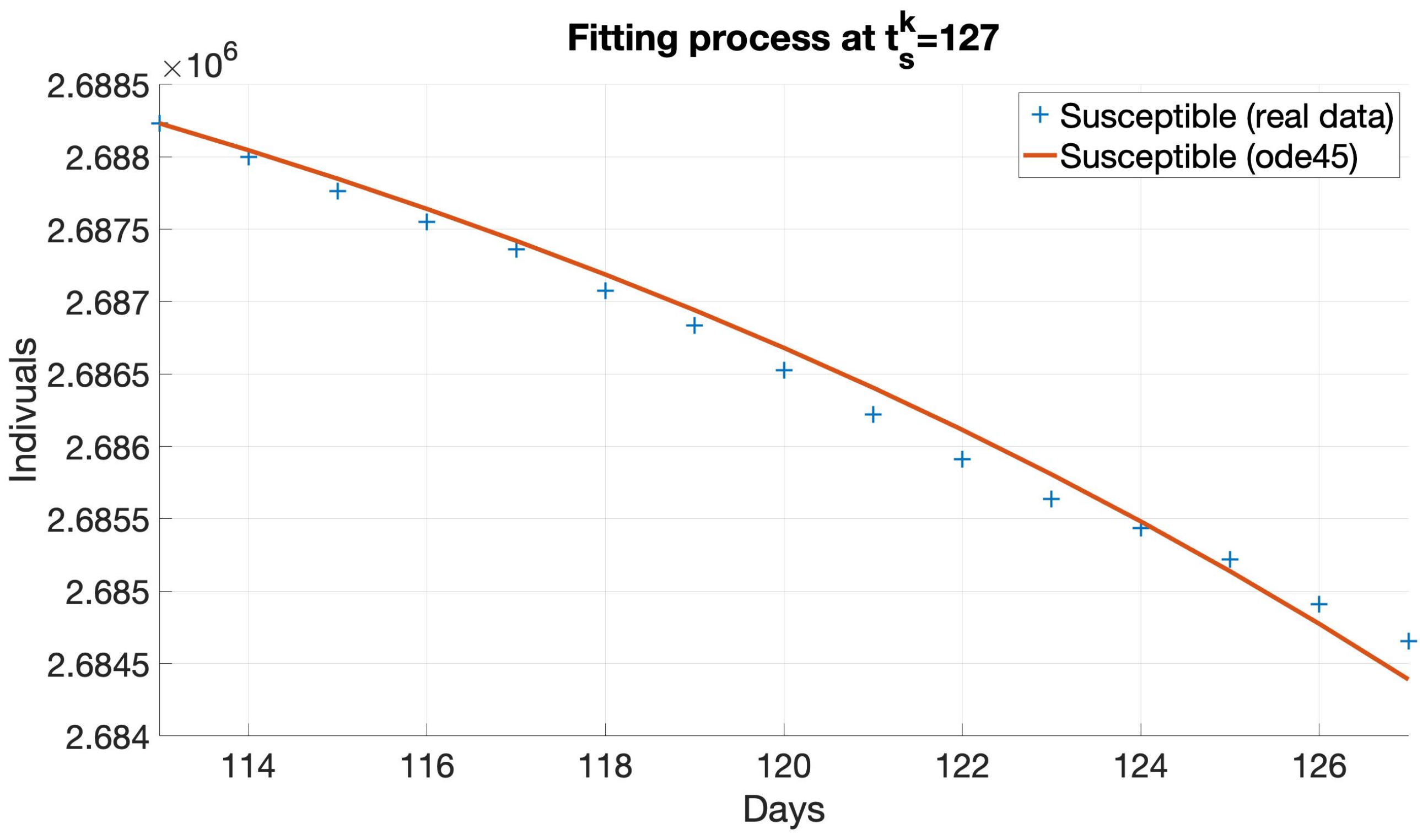}
\caption{Fitting process at $t_s^{17}=127$ (Susceptible Individuals)}
\label{figure3A}
\end{figure}

\begin{figure}[H]
\centering
\includegraphics[width=0.7\linewidth]{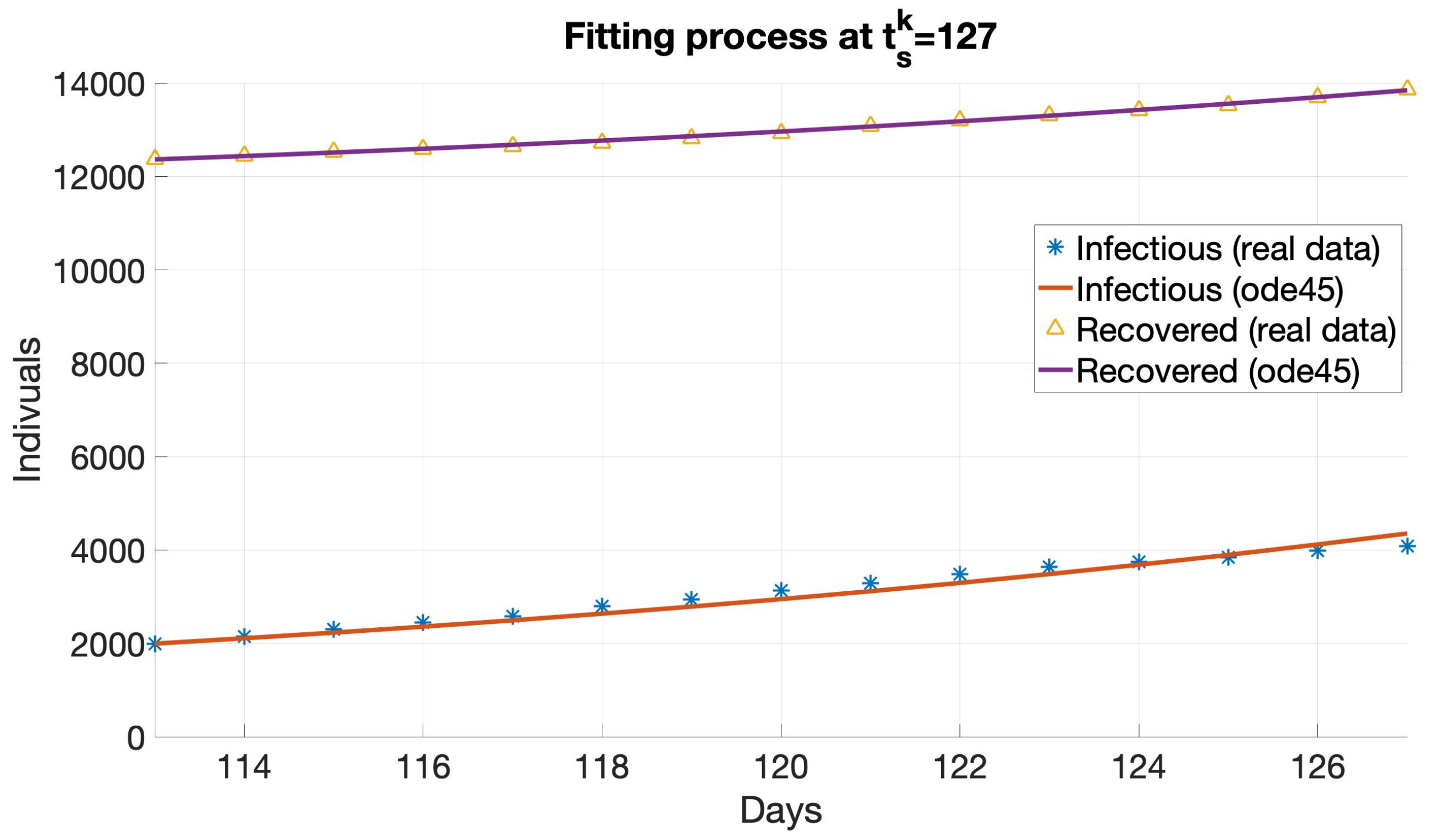}
\caption{Fitting process at $t_s^{17}=127$ (Infectious and Recovered Individuals)}
\label{figure3B}
\end{figure}	

In the next figures we present the evolution of the digital twin (DT) solution to the multivalued problem (\ref{eq:DT2}): in Figure \ref{figure4} we find the evolution of the multivalued solution for the susceptible individuals (note that we have restricted 
the scale for a better display of the results); in figures \ref{figure5} and \ref{figure6} we show the evolution for the infectious and recovered individuals, respectively. Observe that given a time $t\in [t_s^j,t_s^{j+1})$, we have $j+1$  forecasts $\{(S(t,t_s^k),I(t,t_s^k),R(t,t_s^k)),\; k=0,1,\ldots,j\}$. In these figures, the continuous non-emphasized lines correspond to realities that would have occurred if the conditions for the spread of the pandemic had not been modified (confinement, social distancing, use of masks, etc.). In this way, the multivalued solution introduced (Definition 2) implies that the same situation can be reached at a given moment from different previous situations. Of course, some potential realities were avoided due to different measures implemented by the authorities. Now, it is also the case with vaccination against COVID-19. We would also like also to highlight that the basic reproduction number \eqref{eq:r0} changes in the period of analysis, as represented in Figure \ref{figureR0}.

\begin{figure}[H]
\centering
\includegraphics[width=0.7\linewidth]{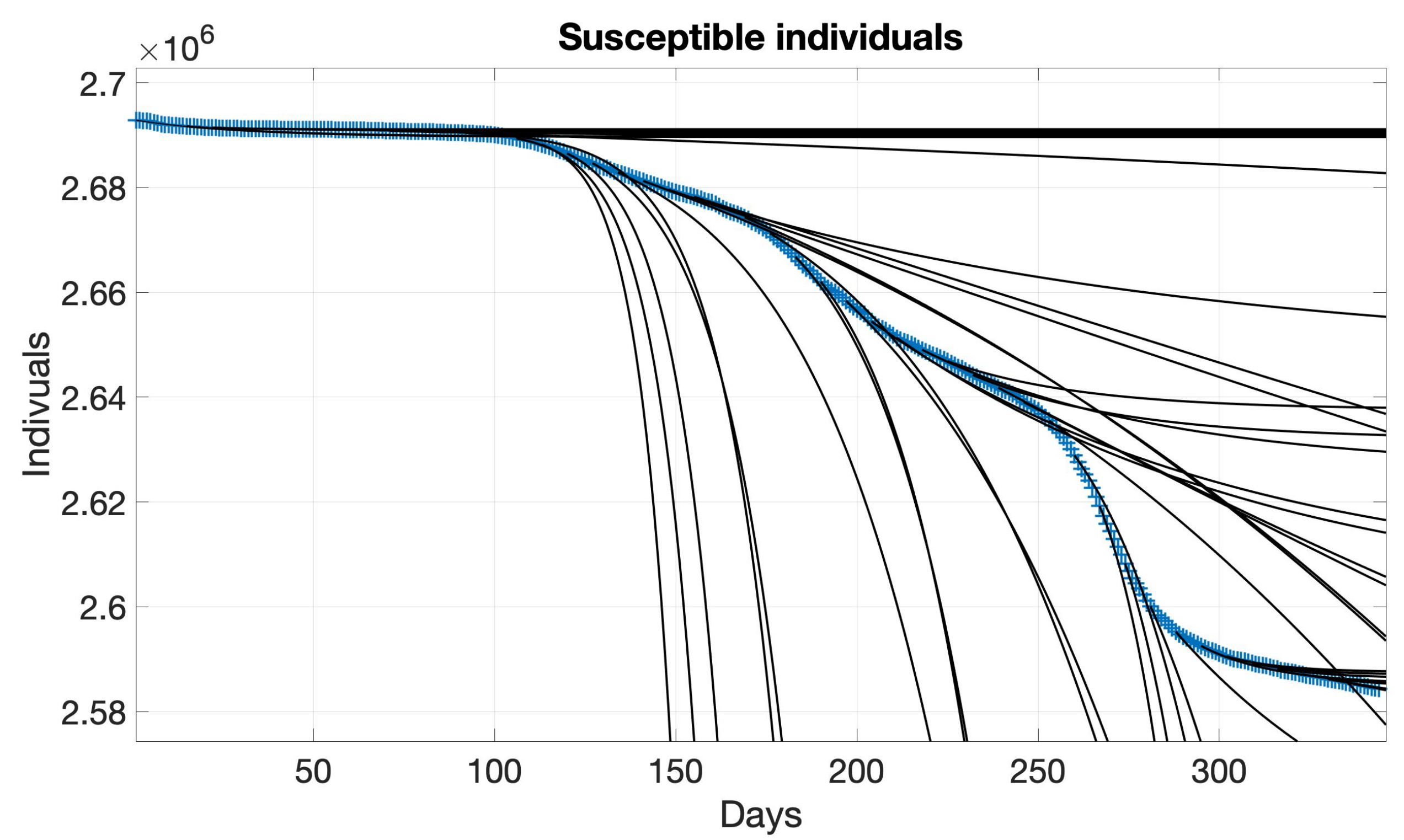}
\caption{Digital twin (DT) solution to the multivalued problem (susceptible individuals) in continuous lines. In discontinuous line the real data}
\label{figure4}
\end{figure}

\begin{figure}[H]
\centering
\includegraphics[width=0.7\linewidth]{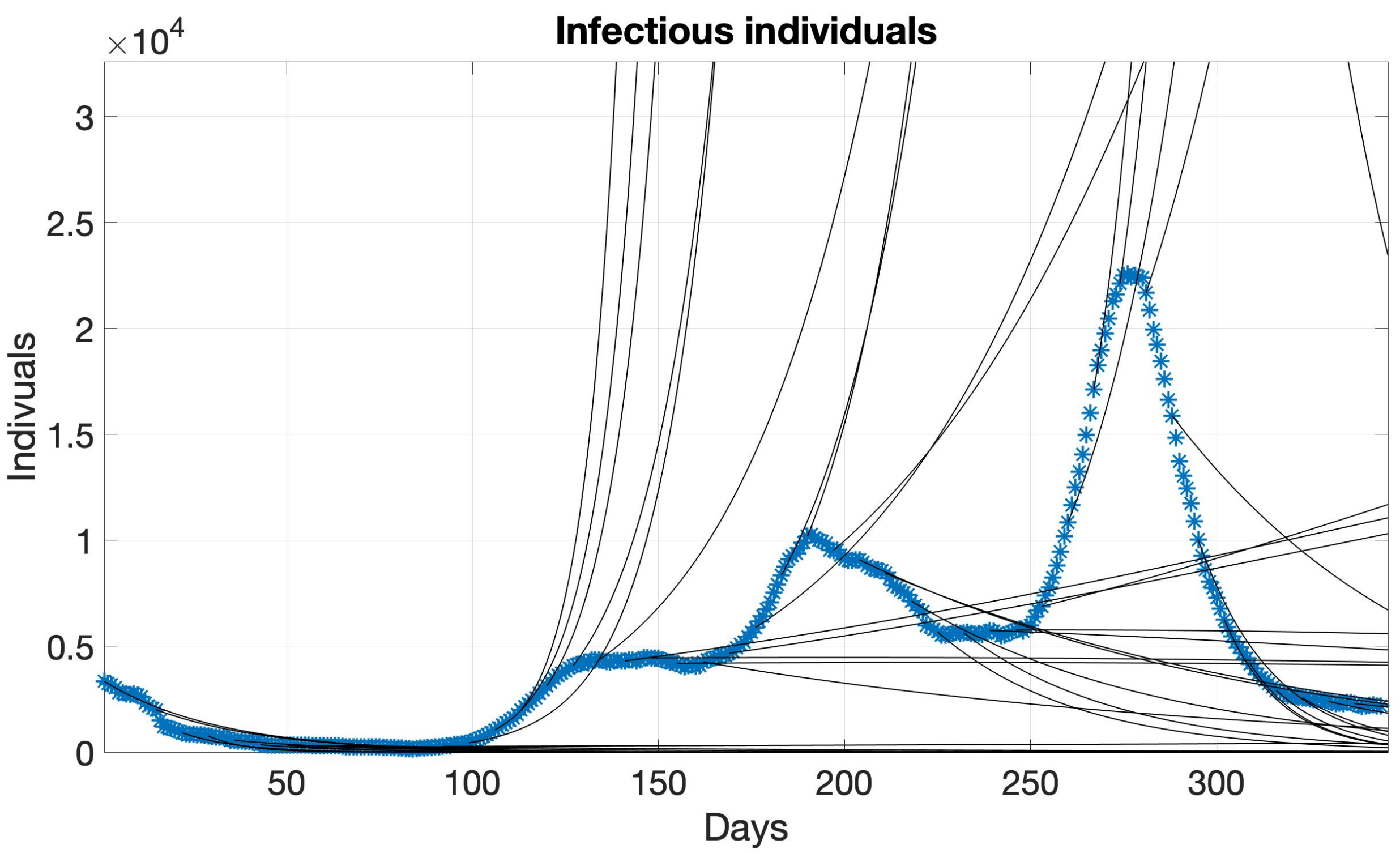}
\caption{Digital twin (DT) solution to the multivalued problem (infectious individuals) in continuous lines. In discontinuous line the real data}
\label{figure5}
\end{figure}

\begin{figure}[H]
\centering
\includegraphics[width=0.7\linewidth]{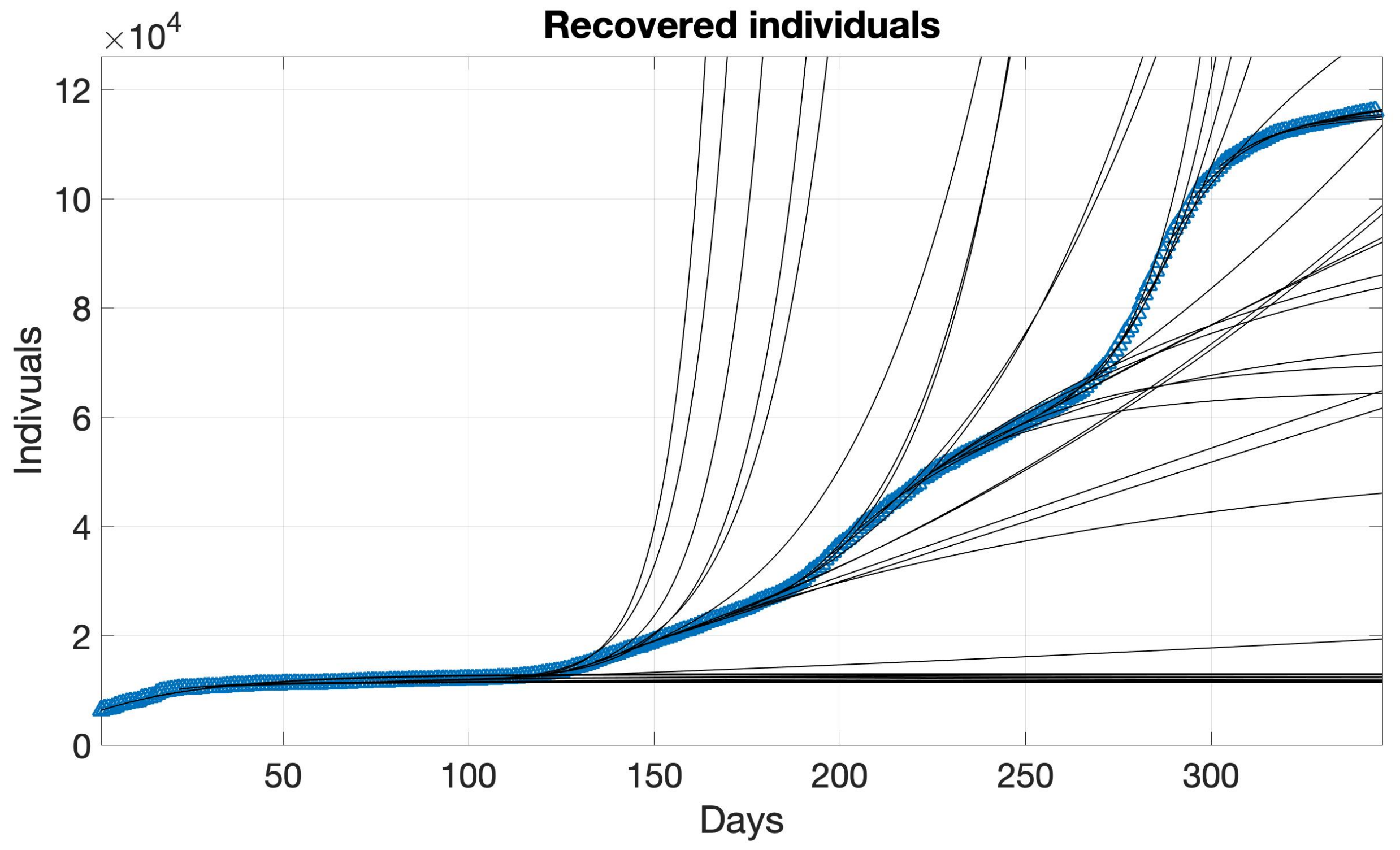}
\caption{Digital twin (DT) solution to the multivalued problem (recovered individuals) in continuous lines. In discontinuous line the real data}
\label{figure6}
\end{figure}

\begin{figure}[H]
\centering
\includegraphics[width=0.7\linewidth]{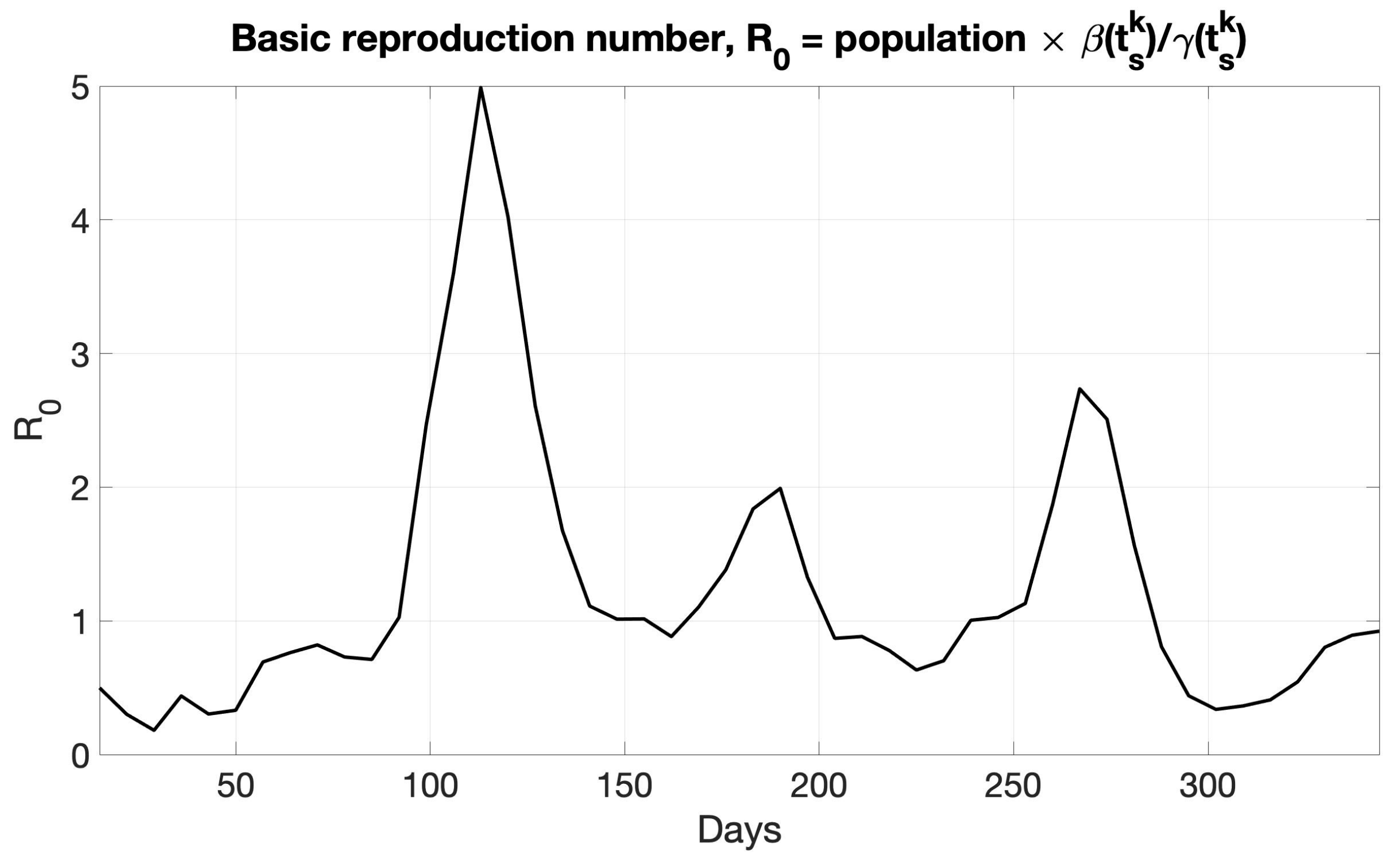}
\caption{Variation of the basic reproduction number \eqref{eq:r0} due to confinement, social distancing, use of masks and/or other policies to stop the spread of the pandemic}
\label{figureR0}
\end{figure}

Now, we will consider the Main Digital Twin (MDT) as the best forecast that we can make at all times \eqref{eq:MDT0}. In Figure \ref{figure7} we represent the evolution of the Main Digital Twin for the susceptible individuals; in Figure \ref{figure8} the evolution 
of MDT associated to the infectious individuals, and in Figure \ref{figure9}, the evolution for the recovered individuals. All of the numerical simulations have been done using the \textbf{ode45} Matlab function. 
The total CPU time required to perform the complete simulation of MDT was $0.18$ seconds.

\begin{figure}[H]
\centering
\includegraphics[width=0.7\linewidth]{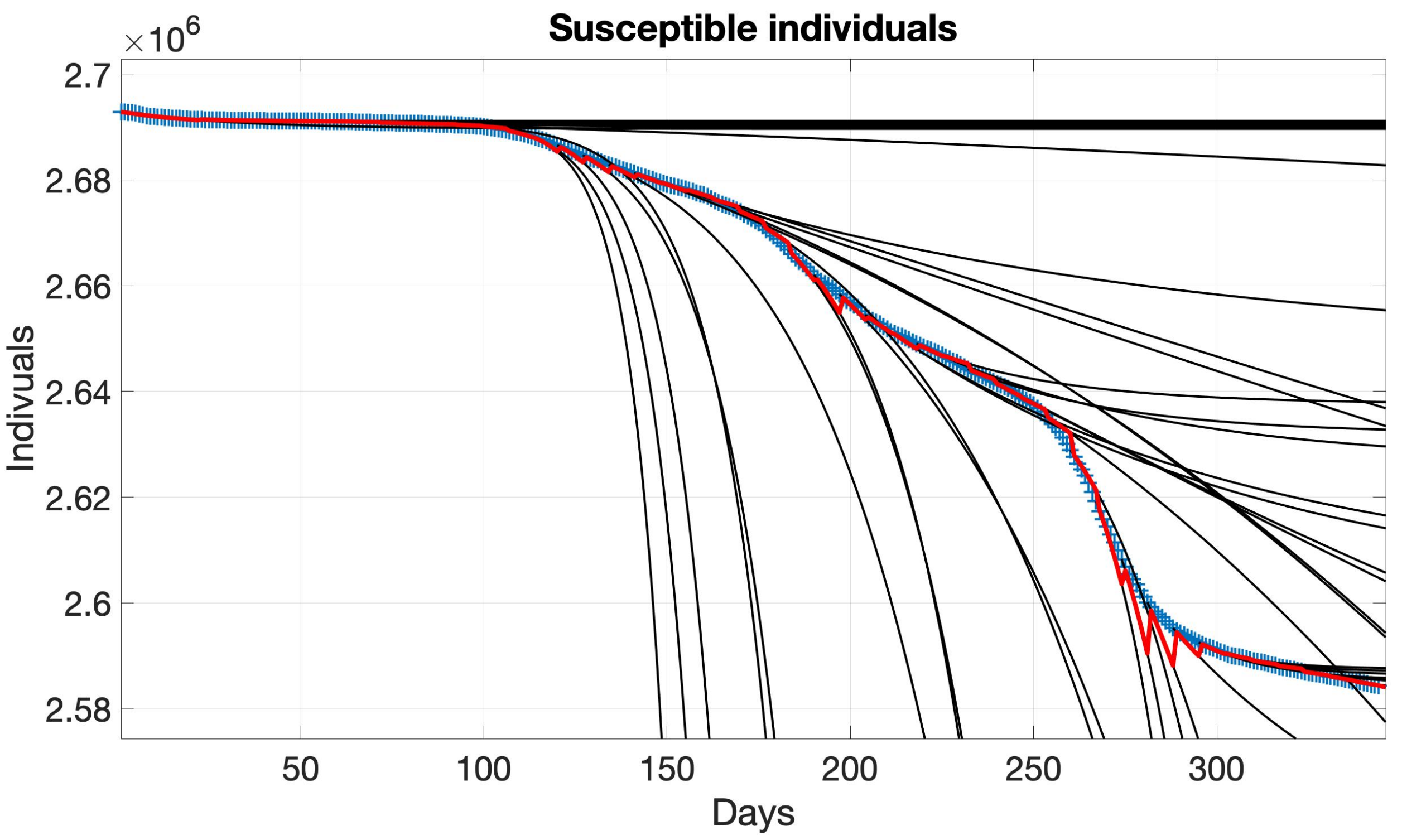}
\caption{Main Digital Twin (MDT) best forecast (susceptible individuals) in red continuous line. Real data in discontinuous line. Digital Twin in continuous black lines}
\label{figure7}
\end{figure}

\begin{figure}[H]
\centering
\includegraphics[width=0.7\linewidth]{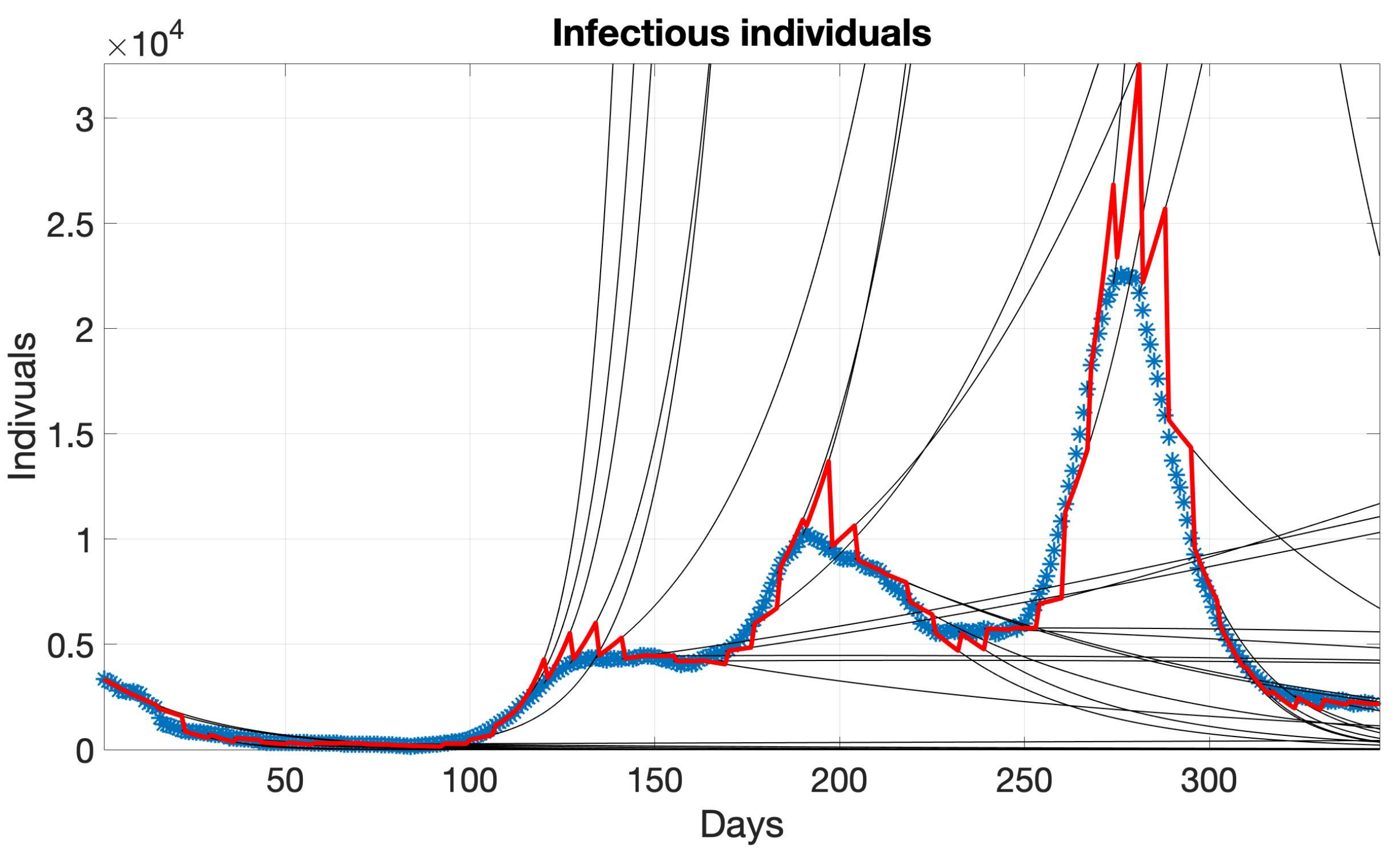}
\caption{Main Digital Twin (MDT) best forecast (infectious individuals)  in red continuous line. Real data in discontinuous line. Digital Twin in continuous black lines}
\label{figure8}
\end{figure}

\begin{figure}[H]
\centering
\includegraphics[width=0.7\linewidth]{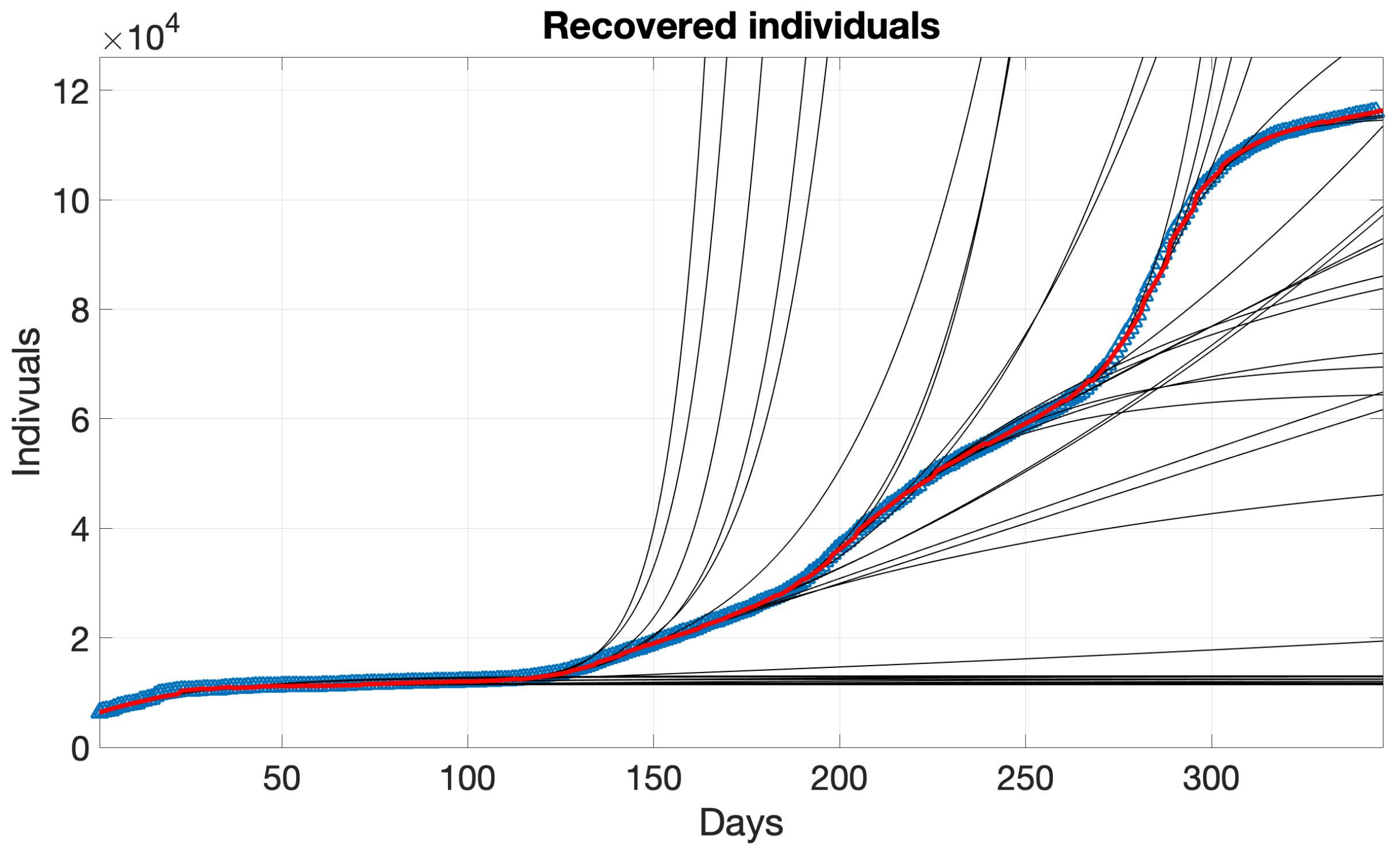}
\caption{Main Digital Twin (MDT) best forecast (recovered individuals) in red continuous line. Real data in discontinuous line. Digital Twin in continuous black lines}
\label{figure9}
\end{figure}

Finally we will compare the Main Digital Twin in terms of definition \eqref{eq:MDT0} with the solution 
of the Stieltjes differential equation \eqref{eq:MDT4}. The numerical approximation of the solution to 
a system of Stieltjes differential equations was introduced in \cite{Fernandez1} where a 
predictor-corrector numerical scheme to approximate the solution to a
Stieljes dieferential equation (also for systems) from a quadrature formula for the Lebesgue
Stieltjes integral was deduced. For this, a finite set of times $\{t_j\}_{j=0}^{N+1}\subset [t_s^0,T]$ 
is considered such that $T_s \subset \{t_j\}_{j=0}^{N+1}$, $t_0=T$, $t_{N+1}=T$ and 
$t_{k+1}-t_k=h>0$, for every $k=1,\ldots,N$. The application of the number scheme in our 
case consists of: given an element $\mathbf{x}_0=(S_0^0,I_0^0,R_0^0)$, we compute 
$\{(\mathbf{x}^+_{j-1},\mathbf{x}^*_j,\mathbf{x}_j)\}_{j=1}^{N+1}$ as
\begin{equation} \label{eq:scheme1}
\left\{\begin{array}{rcl}
x_{i,j}^+ &=&\displaystyle x_{i,k}+F_{s,i}(t_j,\mathbf{x}_j)\, \Delta^+ g_s(t_j), \vspace{0.2cm} \\
x_{i,j+1}^*&=&\displaystyle x_{i,k}^+ + F_{s,i} (t_j^+, \mathbf{x}_j^+ )\,(g_s(t_{j+1})-g_s(t_j^+)), 
\vspace{0.2cm} \\
x_{i,j+1} &=& \displaystyle x_{i,j}^+ + \frac{1}{2} \left(F_{s,i}(t_j^+, \mathbf{x}_j^+ )+
F_{s,i}(t_{j+1}^-,\mathbf{x}_{j+1}^*)\right)\,(g_s(t_{j+1})-g_s(t_j^+)),
\end{array}\right.
\end{equation}
for every $j=0,\ldots,N$ and $i=1,2,3$, being $\mathbf{x}_j=(x_{1,j},x_{2,j},x_{3,j})$. For the 
numerical simulations that we will present below we have considered $h=0.5$. In figures 
\ref{figure10}-\ref{figure12}, the comparison between the numerical solution to \eqref{eq:MDT0} using 
\textbf{ode45} and the solution to \eqref{eq:MDT4} using the scheme \eqref{eq:scheme1} for 
the susceptible, infectious and recovered individuals. Although the calculation 
of the distance between the numerical approximation of \eqref{eq:MDT0} and the 
numerical approximation \eqref{eq:MDT4} is not relevant since both are numerical 
approximations, in the simulations carried out we have observed a maximum difference of 
$1.07e+01$ for the susceptible individuals, $4.08e+00$ for the infectious individuals and 
$2.48e+00$ for the recovered individuals. If we consider $h=0.25$, we obtain a 
maximum difference of $5.56e-01$, $5.65e-01$ and $2.24e+00$ for, respectively, 
the susceptible, infectious and recovered individuals. The CPU time required to perform 
the simulation in the case $h=0.5$ was $0.14$ seconds and $0.15$ seconds in the case 
$h=0.25$.

\begin{figure}[H]
\centering
\includegraphics[width=0.7\linewidth]{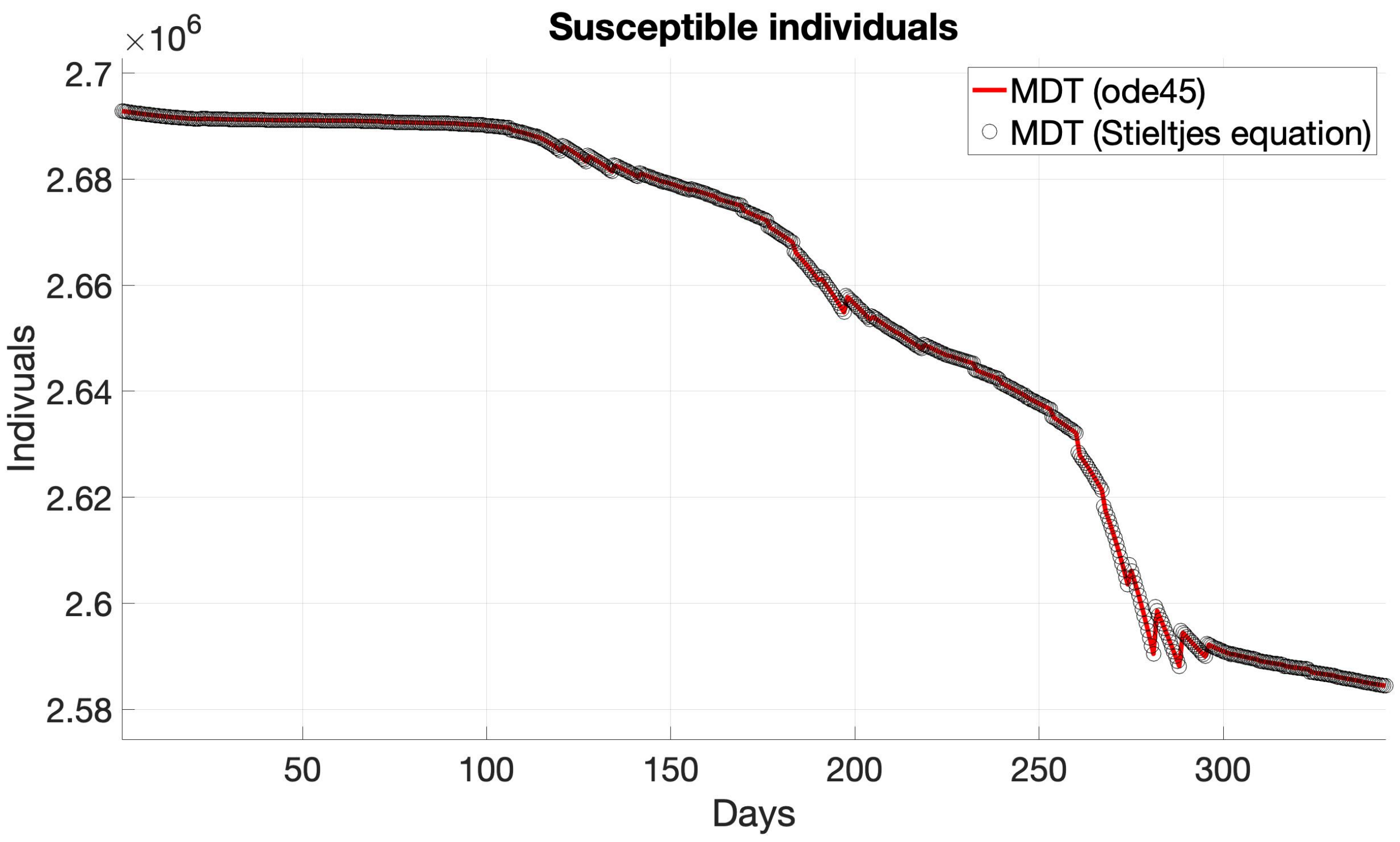}
\caption{Main Digital Twin (MDT) susceptible individuals}
\label{figure10}
\end{figure}

\begin{figure}[H]
\centering
\includegraphics[width=0.7\linewidth]{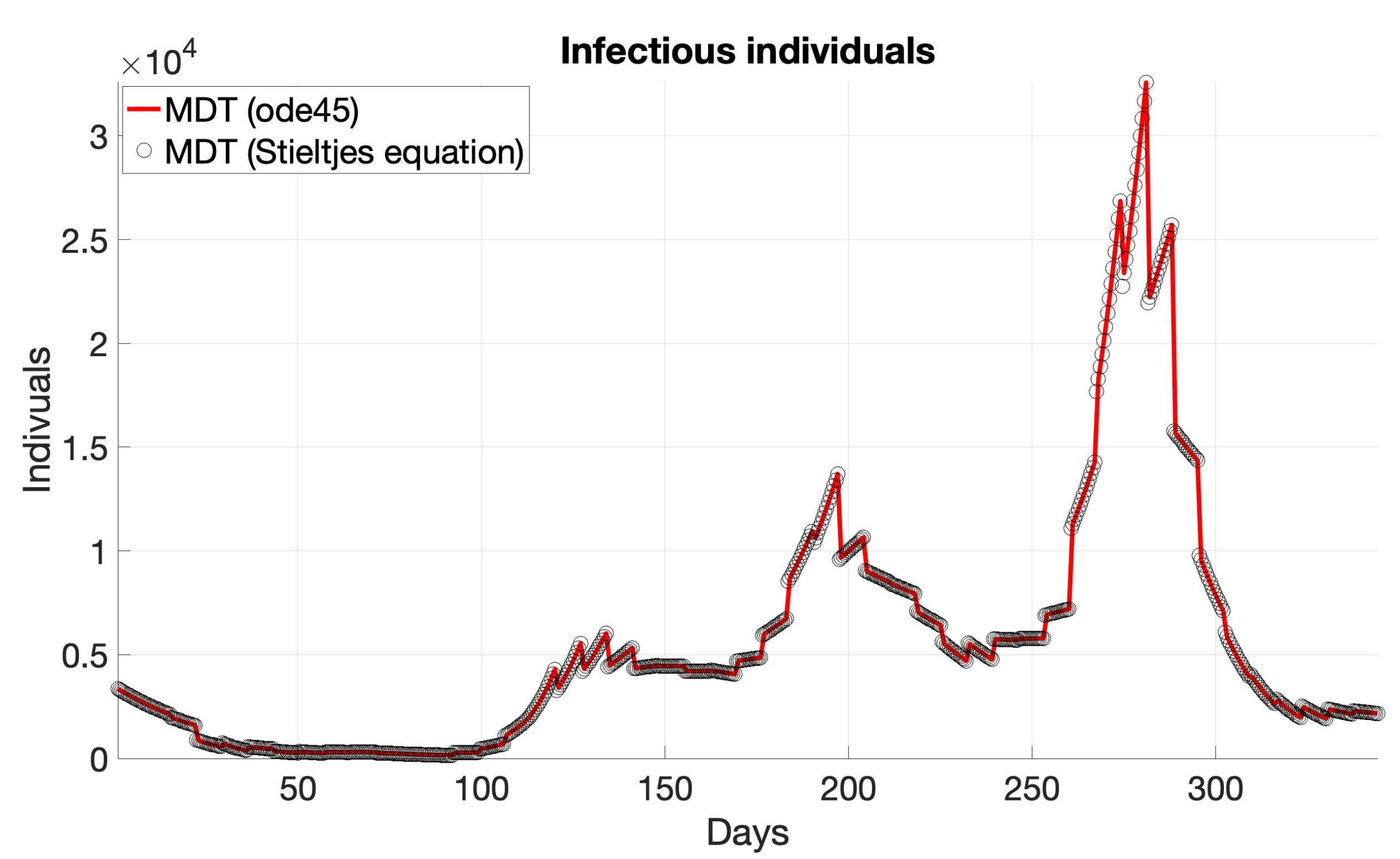}
\caption{Main Digital Twin (MDT) infectious individuals}
\label{figure11}
\end{figure}

\begin{figure}[H]
\centering
\includegraphics[width=0.7\linewidth]{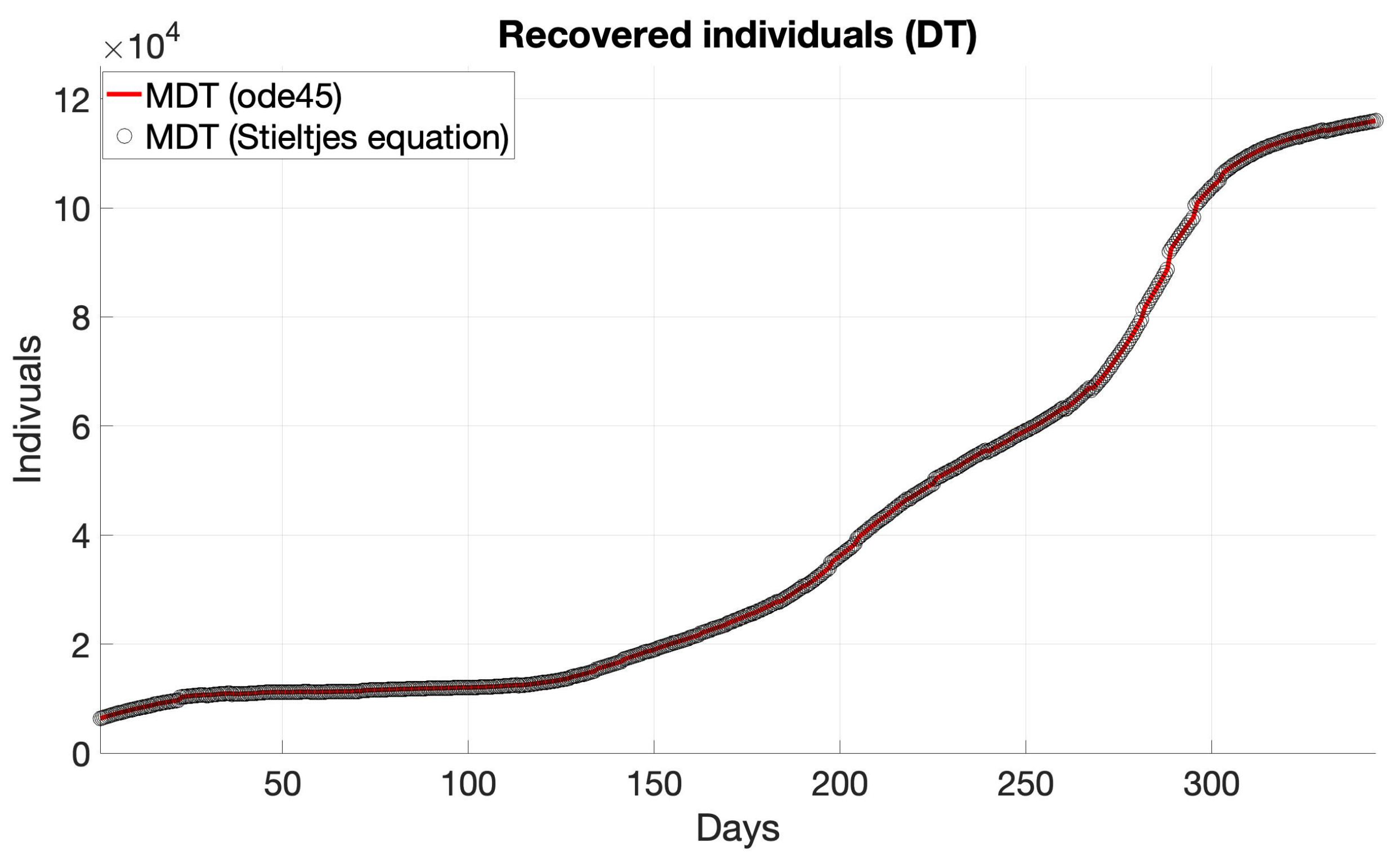}
\caption{Main Digital Twin (MDT) recovered individuals}
\label{figure12}
\end{figure}

\section{Conclusions}\label{sec:5}

The rise of digital twins motivated by industrial requirements puts into focus the need for this relatively new concept to be addressed and formalized in a clear and rigorous way. 
In this letter we propose a digital twin for a SIR mathematical compartmental model. In doing so, Stieltjes derivatives have been used intensively giving rise to a precise definition of solution. We have also analyzed the existence and uniqueness of the solution to the mathematical problem. Moreover, we have applied the proposed ideas to a specific problem: the spread of the pandemic of COVID-19 in Galicia. The numerical experiments performed show that the proposed definitions and results can be used to explain other pandemics.

\section*{Competing interests}

The authors declare that they have no known competing financial interests or personal relationships that could have appeared to influence the work reported in this paper.

\section*{Acknowledgments}
This work has been partially supported by the Agencia Estatal de Investigaci\'on (AEI) of Spain under Grant MTM2016-75140-P, cofinanced by the European Community fund FEDER, as well as by Instituto de Salud Carlos III, grant COV20/00617.  FJF, FAFT and JJN are beneficiary of Xunta de Galicia grant ED431C 2019/02 for Competitive Reference Research Groups (2019-22).

\end{document}